\documentclass{article}

\usepackage[a4paper,left=1.0in,right=1.0in,top=1.1in,bottom=1.3in]{geometry}

\usepackage{graphicx}
\usepackage{amsthm,amsmath,amssymb}
\usepackage{enumitem}
\usepackage{thmtools}
\usepackage{comment}
\usepackage{thm-restate}
\usepackage{color}
\usepackage[noadjust]{cite}
\usepackage{multirow}
\usepackage{tabu,array,soul}
\usepackage{algpseudocode}
\usepackage{siunitx}
\usepackage{textcomp}
\usepackage{complexity}
\usepackage[framemethod=TikZ]{mdframed}
\usepackage{hyperref}

\usetikzlibrary{chains,fit}
\usetikzlibrary{shapes,arrows}
\usetikzlibrary{shapes.geometric}
\usetikzlibrary{matrix,positioning,calc}
\usetikzlibrary{decorations.markings}
\usetikzlibrary{decorations.pathmorphing}
\usetikzlibrary{decorations.pathreplacing}

\tikzstyle{vertex}=[circle, draw, inner sep=0pt, minimum size=6pt]
\tikzstyle{vertbox}=[draw, inner sep=0pt, minimum size=8pt]

\usepackage{setspace}
\usepackage{microtype}
\usepackage{authblk}
\usepackage{xcolor}
\usepackage{etex,ifthen}
\usepackage{etoolbox}
\usepackage{rotating}
\usepackage{tikz}
\usepackage{pgfplots}\pgfplotsset{compat=1.14}
\usepackage{tikz-cd}
\usepackage{url}
\usepackage{calc}
\usepackage[font={small,it}]{caption}
\usepackage{alphalph}
\usepackage{amsthm}
\usepackage{eqparbox}
\usepackage{afterpage}
\usepackage{amsfonts,amsmath,amssymb}
\usepackage{mathrsfs}
\usepackage{algorithm}
\usepackage{bm}
\usepackage{listings}
\usepackage[english]{babel}
\usepackage{verbatim}
\usepackage{mathptmx}
\usepackage{appendix}
\usepackage[all]{xy}

\usepackage{graph}

\usetikzlibrary{chains,fit}
\usetikzlibrary{shapes.geometric}
\usetikzlibrary{matrix,positioning,calc}
\usetikzlibrary{decorations.markings}
\usetikzlibrary{decorations.pathmorphing}
\tikzstyle{vertex}=[circle, draw, inner sep=0pt, minimum size=6pt]
\tikzstyle{svertex}=[circle, draw, inner sep=0pt, minimum size=3pt]
\tikzstyle{dvertex}=[circle, draw, inner sep=0pt, minimum size=9pt]
\tikzstyle{vertbox}=[draw, inner sep=0pt, minimum size=8pt]

\newcommand{\oset}[3][0ex]{
    \mathrel{\mathop{#3}\limits^{
		\vbox to#1{\kern-2\ex@
		\hbox{$\scriptstyle#2$}\vss}}}}

\newcommand{\cG}{\mathcal{G}}

\mdfdefinestyle{MyFrame}{%
    linecolor=black,
    innerbottommargin=\baselineskip,
    backgroundcolor=gray!10!white}

\newcommand{\etal}{\textit{et al}.}
\newcommand{\ie}{\textit{i}.\textit{e}.}
\newcommand{\eg}{\textit{e}.\textit{g}.}

\newcommand{\PureTwoDir}{\textup{\textsc{Pure-2-Dir}}}
\newcommand{\String}{\textup{\textsc{String}}}
\newcommand{\OneString}{\textup{\textsc{1-String}}}

\newcommand{\PHPCa}{\textup{\textsc{Planar Hamiltonian}}} \newcommand{\PHPCb}{\textup{\textsc{Path Completion}}}

\newcommand{\PHPC}{\PHPCa\ \PHPCb}

\newcommand{\PHCC}{\textup{\textsc{Planar Hamiltonian Cycle Completion}}}

\newcommand{\apexgraph}[1]{#1_{\mathsf{apex}}}
\newcommand{\Grep}[1]{#1_{\mathsf{rep}}}

\newcommand{\area}{\mathsf{area}}
\newcommand{\rank}{\mathsf{rank}}
\newcommand{\xpos}{\mathsf{xpos}}

\newcommand{\ellfront}{\ell_{\mathsf{front}}}

\newcommand{\Eleft}[1]{#1_{\mathsf{left}}}
\newcommand{\Eleftedge}{\Eleft{E}}
\newcommand{\Eright}[1]{#1_{\mathsf{right}}}
\newcommand{\Erightedge}{\Eright{E}}
\newcommand{\Ecross}[1]{#1_{\mathsf{cross}}}
\newcommand{\Ecrossedge}{\Ecross{E}}
\newcommand{\Eabove}[1]{#1_{\mathsf{above}}}
\newcommand{\Eaboveedge}{\Eabove{E}}
\newcommand{\Ebelow}[1]{#1_{\mathsf{below}}}
\newcommand{\Ebelowedge}{\Ebelow{E}}

\newcommand{\subdivide}[1]{#1_{k\operatorname{-{\mathsf{div}}}}}

\newcommand{\planar}[1]{#1_{\mathsf{pl}}}
\newcommand{\draw}{\textsc{Draw}}

\renewcommand{\curve}[1]{\mathtt{c}(#1)}

\declaretheorem[numberlike=equation]{Theorem}
\declaretheorem[numberlike=equation]{Lemma}
\declaretheorem[numberlike=equation]{Corollary}

\declaretheoremstyle[bodyfont=\it,qed=$\lozenge$]{defstyle}
\declaretheorem[numberlike=equation,style=defstyle]{Definition}

\declaretheorem[numberlike=equation]{Observation}

\makeatletter
\patchcmd{\ALG@step}{\addtocounter{ALG@line}{1}}{\refstepcounter{ALG@line}}{}{}
\newcommand{\ALG@lineautorefname}{Line}
\makeatother

\addto\extrasenglish{%
}

\hypersetup{
	colorlinks,
	linkcolor=red!75!black,
	citecolor=blue!75!black,
	urlcolor=green!75!black
}

\begin{document}

	\title{Finding Geometric Representations of Apex Graphs is NP-Hard}
	
    %\date{}
    
    \author{Dibyayan Chakraborty\thanks{Indian Institute of Science, Bengaluru. E-mail: \texttt{dibyayancg@gmail.com}}~ and Kshitij Gajjar\thanks{National University of Singapore, Singapore. A part of this work was done when the author was a postdoctoral researcher at Technion, Israel. E-mail: \texttt{kshitijgajjar@gmail.com}}}
    
    \maketitle
    
    \setstretch{1.1}

\begin{abstract}
Planar graphs can be represented as intersection graphs of different types of geometric objects in the plane, \eg, circles (Koebe, 1936), line segments (Chalopin \& Gon{\c{c}}alves, 2009), \textsc{L}-shapes (Gon{\c{c}}alves~\etal, 2018). 
% Furthermore, these representations can be obtained in polynomial time when the planar graph is provided as input. 
For general graphs, however, even deciding whether such representations exist is often $\NP$-hard. We consider apex graphs, \ie, graphs that can be made planar by removing one vertex from them. 
% For this reason, apex graphs are sometimes also called almost planar graphs. 
We show, somewhat surprisingly, that deciding whether geometric representations exist for apex graphs is $\NP$-hard.

More precisely, we show that for every positive integer $k$, recognizing every graph class $\mathcal{G}$ which satisfies $\PureTwoDir\subseteq\mathcal{G}\subseteq\OneString$ is $\NP$-hard, even when the input graphs are apex graphs of girth at least $k$. Here, $\PureTwoDir$ is the class of intersection graphs of axis-parallel line segments (where intersections are allowed only between horizontal and vertical segments) and $\OneString$ is the class of intersection graphs of simple curves (where two curves share at most one point) in the plane. This partially answers an open question raised by Kratochv{\'\i}l \& Pergel (2007).

Most known $\NP$-hardness reductions for these problems are from variants of 3-SAT. We reduce from the $\PHPC$ problem, which uses the more intuitive notion of planarity. As a result, our proof is much simpler and encapsulates several classes of geometric graphs.

\end{abstract}

\section{Introduction}

The recognition a graph class is the decision problem of determining whether a given simple, undirected, unweighted graph belongs to the graph class. Recognition of graph classes is a fundamental research topic in discrete mathematics with applications in VLSI design~\cite{chung1983diogenes,chung1984,sherwani2007}. In particular, when the graph class relates to intersection patterns of geometric objects, the corresponding recognition problem finds usage in disparate areas like map labelling~\cite{agarwal1998label}, wireless networks~\cite{kuhn2008ad}, and computational biology~\cite{xu2006fast}. The applicability of the recognition problem further increases if the given graph is planar.

Several $\NP$-hard graph problems admit efficient algorithms for planar graphs~\cite{stockmeyer1973,hopcroft1974,hadlock1975finding}. There are many reasons behind this, owing to the nice structural properties of planar graphs~\cite{tutte1963draw,lipton1979separator}. Since most of the graphs are non-planar (\ie, the fraction of $n$-vertex graphs that are planar approaches zero as $n$ tends to infinity), it is worthwhile to explore whether these efficient algorithms continue to hold for graphs that are ``close to'' planar. Specifically, we deal with apex graphs.

\begin{Definition}\cite{welsh1993,thilikos1997}. A graph is an apex graph if it contains a vertex whose removal makes it planar.
\end{Definition}

Since every $n$-vertex apex graph contains an $(n-1)$-vertex planar graph as an induced subgraph, one might expect apex graphs to retain some of the properties of planar graphs, in the hope that efficient algorithms for planar graphs carry forward to apex graphs as well. In this paper, we show that for recognition problems for several classes of geometric intersection graphs, this is not the case.

The intersection graph of a collection $C$ of sets is a graph with vertex set $C$ in which two vertices are adjacent if and only if their corresponding sets have a non-empty intersection. For our main result, we are particularly interested in two natural and well-studied classes of geometric intersection graphs (\ie, when the sets in the collection $C$ are geometric objects in the plane) called $\PureTwoDir$ and $\OneString$, shown\footnote{This illustration is only for representational purposes.} in~\autoref{fig:P2DG1S}.

\begin{Definition} \label{def:p2d}
	$\PureTwoDir$ is the class of all graphs $G$, such that $G$ is the intersection graph of axis-parallel line segments in the plane, where intersections are allowed only between horizontal and vertical segments.
\end{Definition}

\begin{Definition} \label{def:1str}
	$\OneString$ is the class of all graphs $G$, such that $G$ is the intersection graph of simple curves\footnote{Formally, a simple curve is a subset of the plane which is homeomorphic to the interval $[0, 1]$.} in the plane, where two intersecting curves share exactly one point, at which they must cross each other.
\end{Definition}

\begin{mdframed}[style=MyFrame]
	%\vspace{0.5em}
	\begin{Theorem}[Main Result]\label{thm:main}
		Let $g$ be a positive integer and $\cG$ be a graph class such that $$\PureTwoDir \subseteq \cG \subseteq \OneString.$$ Then it is $\NP$-hard to decide whether an input graph belongs to $\cG$, even when the input graphs are restricted to bipartite apex graphs of girth at least $g$.
	\end{Theorem}
	\vspace{-0.5em}
\end{mdframed}

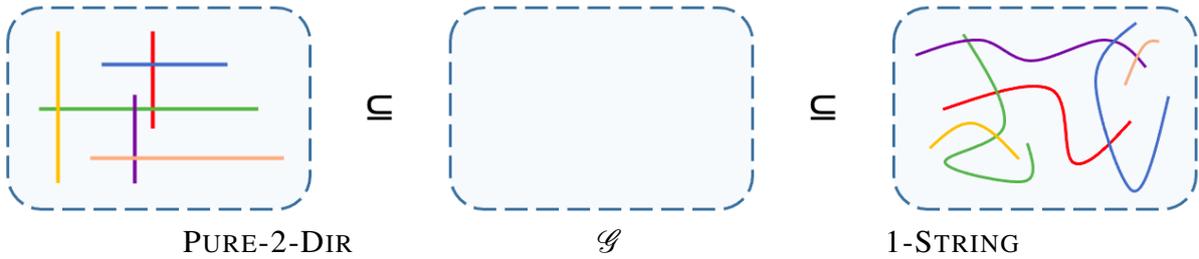
\begin{figure}[h]
%\vspace{0.5em}
\centering
\includegraphics[width=1\textwidth]{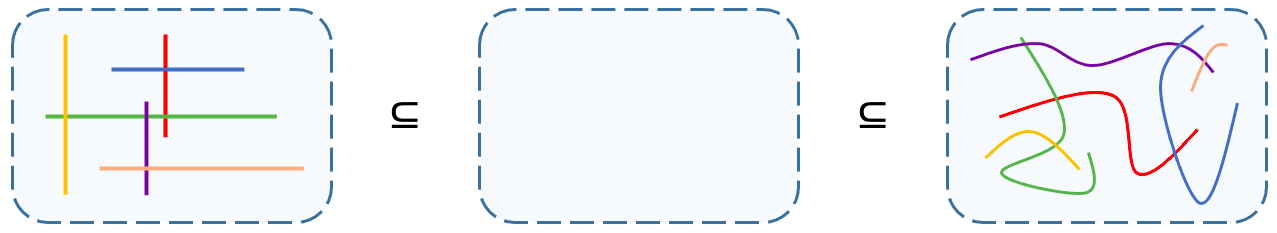}
\begin{tikzpicture}
\node at (-4.7,0) {\large $\PureTwoDir$};
\node at (-0.2,0) {\large $\cG$};
\node at (4.3,0) {\large $\OneString$};
\end{tikzpicture}
\caption{A visual depiction of our main result (\autoref{thm:main}).%Note that in the $\PureTwoDir$ (\autoref{def:p2d}) representation, two vertical (similarly, two horizontal) line segments do not intersect each other. Also note that in the $\OneString$ (\autoref{def:1str}) representation, two intersecting curves share exactly one point, and the curves do not touch (tangentially) but instead cross each other at that point. Our main result states that the recognition of every graph class $\cG$ that lies between these two is $\NP$-hard, even if the input graph is bipartite, apex, and has large girth. Some candidate graph classes for $\cG$ are shown in~\autoref{sec:significant}.
}\label{fig:P2DG1S}
%\vspace{0.5em}
\end{figure}

In~\autoref{fig:P2DG1S}, note that in the $\PureTwoDir$ (\autoref{def:p2d}) representation, two vertical (similarly, two horizontal) line segments do not intersect each other. Also note that in the $\OneString$ (\autoref{def:1str}) representation, two intersecting curves share exactly one point, and the curves do not touch (tangentially) but instead cross each other at that point. Our main result states that the recognition of every graph class $\cG$ that lies between these two is $\NP$-hard, even if the input graph is bipartite, apex, and has large girth. Some candidate graph classes for $\cG$ are shown in~\autoref{sec:significant}.

\medskip \noindent \textbf{Paper Roadmap:} In~\autoref{sec:significant}, we highlight some implications of our main result and in~\autoref{sec:related}, we survey the existing literature on the topic. In~\autoref{sec:proof-technique}, we describe our proof techniques and give an overview of our proof. In~\autoref{sec:proof}, we prove~\autoref{thm:main} and draw some conclusions in~\autoref{sec:conclude}.

\subsection{Significance of the Main Result}\label{sec:significant}

Our main result has several corollaries, obtained by substituting different values for the graph class $\cG$. Recall that the recognition of a graph class $\mathcal{G}$ asks if a given graph $G$ is a member of $\mathcal{G}$. 

$\String$ is the class of intersection graphs of simple curves in the plane. Kratochv{\'\i}l \& Pergel~\cite{kratochvil2007geometric} posed the question of determining the complexity of recognizing $\String$ when the inputs are restricted to graphs with large girth. The above question was answered by Musta{\c{t}}{\u{a}} \& Pergel~\cite{mustac2019}, where they showed that recognizing $\String$ is $\NP$-hard, even when the inputs are restricted to graphs of arbitrarily large girth. However, the graphs they constructed were far from planar. Since $\OneString \subsetneq \String$, the following corollary of our main result partially answers Kratochv{\'\i}l \& Pergel's~\cite{kratochvil2007geometric} question when the inputs are restricted to apex graphs of large girth.

\begin{Corollary}
	For every positive integer $g$, recognizing $\OneString$ is $\NP$-hard, even for bipartite apex graphs with girth at least $g$.
\end{Corollary}

%The following two corollaries simplify and strengthen several results of Kratochv{\'\i}l~\cite{kratochvil1991,kratochvil1994}, Kratochv{\'\i}l \& Matou{\v{s}}ek~\cite{kratochvil1989}, Kratochv{\'\i}l \& Pergel~\cite{kratochvil2007geometric}.
Chalopin \& Gon{\c{c}}alves~\cite{chalopin2009scheinerman} showed that every planar graph can be represented as an intersection graph of line segments in polynomial time. The following corollary shows that a similar result does not hold for apex graphs.

\begin{Corollary} \label{cor:lineseg}
	For every positive integer $g$, recognizing intersection graphs of line segments is $\NP$-hard, even for bipartite apex graphs with girth at least $g$.
\end{Corollary}

Gon{\c{c}}alves, Isenmann \& Pennarun~\cite{gonccalves2018Lshapes} showed that every planar graph can be represented as an intersection graph of \textsc{L}-shapes in polynomial time. The following corollary shows that a similar result does not hold for apex graphs.

\begin{Corollary} \label{cor:lshapes}
	For every positive integer $g$, recognizing intersection graphs of \textsc{L}-shapes is $\NP$-hard, even for bipartite apex graphs with girth at least $g$.
\end{Corollary}

Our main result also has a connection to a graph invariant called~\emph{boxicity}. The boxicity of a graph is the minimum integer $d$ such that the graph can be represented as an intersection graph of $d$-dimensional axis-parallel boxes. Thomassen showed three decades ago that the boxicity of every planar graph is either one, two or three~\cite{thomassen1986interval}. It is easy to check if the boxicity of a planar graph is one~\cite{booth1976testing}. However, the complexity of determining whether a planar graph has boxicity two or three is not yet known. A result of Hartman, Newman \& Ziv~\cite{hartman1991} states that the class of bipartite graphs with boxicity 2 is precisely $\PureTwoDir$. Combined with our main result, this implies that determining the boxicity of apex graphs is $\NP$-hard.

\begin{Corollary} \label{cor:box}
	For every positive integer $g$, recognizing graphs with boxicity $2$ is $\NP$-hard, even for bipartite apex graphs with girth at least $g$.
\end{Corollary}

%\autoref{cor:lineseg} and~\autoref{cor:box} partially addresses another open question posed by Kratochv{\'\i}l \& Pergel~\cite{kratochvil2007geometric}.
\textsc{Conv} is the class of intersection graphs of convex objects in the plane. Kratochv{\'\i}l \& Pergel~\cite{kratochvil2007geometric} asked if recognizing~\textsc{Conv} remains $\NP$-hard when the inputs are restricted to graphs with large girth. Note that the class of graphs with boxicity $2$ (alternatively, intersection graphs of rectangles) is a subclass of~\textsc{Conv}. Similarly, intersection graphs of line segments on the plane is also a subclass of~\textsc{Conv}. Hence,~\autoref{cor:lineseg} and~\autoref{cor:box} also partially address the aforementioned open question of Kratochv{\'\i}l \& Pergel~\cite{kratochvil2007geometric}.

A graph is $c$-apex if it contains a set of $c$ vertices whose removal makes it planar. This is a natural generalization of apex graphs. Our main result implies that no graph class $\cG$ satisfying $\PureTwoDir\subseteq\cG\subseteq\OneString$ can be recognized in $n^{f(c)}$ time, where $f$ is a computable function depending only on $c$. This means recognizing $\cG$ is $\XP$-hard, and thus not fixed-parameter tractable~\cite{niedermeier2006invitation,downey2012parameterized} for $c$-apex graphs when parameterized by $c$.

\begin{Corollary} \label{cor:fpt}
	Let $g$ be a positive integer and $\cG$ be a graph class such that $$\PureTwoDir\subseteq\cG\subseteq\OneString.$$ Then assuming $\P\neq\NP$, there is no $f(c)\cdot n^{O(1)}$ time algorithm that recognizes $\cG$ (where $f$ is a computable function depending only on $c$), even for bipartite $c$-apex graphs with girth at least $g$.
\end{Corollary}

Thanks to a long line of work due to Robertson \& Seymour~\cite{robertsonseymour}, several graph classes can be characterized by a finite set of forbidden minors. For example, planar graphs are $(K_5,K_{3,3})$-minor free graphs. Interestingly, the set of forbidden minors is not known for apex graphs, although it is known that the set is finite~\cite{gupta1991}. However, it is easy to see that apex graphs are $K_6$-minor free, which means that our main result has the following implication.

\begin{Corollary} \label{cor:minorfree}
	Let $g$ be a positive integer and $\cG$ be a graph class such that $$\PureTwoDir\subseteq\cG\subseteq\OneString.$$ Then it is $\NP$-hard to decide whether an input graph belongs to $\cG$, even for bipartite $K_6$-minor free graphs with girth at least $g$.
\end{Corollary}

Finally, using techniques different from ours, Kratochv{\'\i}l \& Matou{\v{s}}ek~\cite{kratochvil1989} had shown that recognizing $\PureTwoDir$ is $\NP$-hard, and so is the recognition of line segment intersection graphs. \autoref{thm:main} and~\autoref{cor:lineseg} show that these recognition problems remain $\NP$-hard even when the inputs are restricted to bipartite apex graphs of arbitrarily large girth, thereby strengthening their results.

\subsection{Related Work}\label{sec:related}

In this paper, we focus on recognition problems for subclasses of string graphs. A graph is called a string graph if it is the intersection graph of a set of simple curves in the plane. Benzer~\cite{benzer1959} initiated the study of string graphs over half a century ago. Sinden, in his seminal paper~\cite{sinden1966}, asked whether the recognition of string graphs is decidable. Kratochv{\'\i}l~\cite{kratochvil1991} took the first steps towards answering Sinden's question by showing that recognizing string graphs is $\NP$-hard. Schaefer, Sedgwick \& {\v{S}}tefankovi{\v{c}},~\cite{schaefer2003} settled Sinden's question by showing that the recognition of string graphs also lies in $\NP$, and is thus $\NP$-complete.

Finding representations of planar graphs using geometric objects has been a very popular and exciting line of research. %Sinden~\cite{sinden1966} proved that all planar graphs are string graphs.
The celebrated Circle Packing Theorem~\cite{koebe1936} (also see~\cite{andreev1970convex,thurston1982hyperbolic}) states that all planar graphs can be expressed as intersection graphs of touching disks. This result also implies that planar graphs are string graphs. In his PhD thesis, Scheinerman~\cite{scheinerman1984} conjectured that planar graphs can be expressed as intersection graphs of line segments. Scheinerman's conjecture was first shown to be true for bipartite planar graphs~\cite{hartman1991}, and later extended to triangle-free planar graphs~\cite{de1999}. %Researchers also tried to prove relaxed versions of Scheinerman's conjecture.
Chalopin, Gon{\c{c}}alves \& Ochem~\cite{chalopin2010onestring} proved a relaxed version of Scheinerman's conjecture by showing that every planar graph is in $\OneString$. Scheinerman's conjecture was finally proved in 2009 by Chalopin \& Gon{\c{c}}alves~\cite{chalopin2009scheinerman}. More recently, Gon{\c{c}}alves, Isenmann \& Pennarun~\cite{gonccalves2018Lshapes} showed that planar graphs are intersection graphs of \textsc{L}-shapes. Soon thereafter, Gon{\c{c}}alves, L{\'e}v{\^e}que \& Pinlou~\cite{gonccalves2019homothetic} showed that planar graphs are also
intersection graphs of homothetic triangles.

Another interesting subclass of string graphs is \textsc{$k$-Dir}~\cite{cabello2017refining}. A graph $G$ is in \textsc{$k$-Dir} (and has a \textsc{$k$-Dir} representation) if $G$ is the intersection graph of line segments whose slopes belong to a set of at most $k$ real numbers. Moreover, $G$ is in \textsc{Pure-$k$-Dir} if $G$ has a \textsc{$k$-Dir} representation where no two segments of the same slope intersect. West~\cite{west} conjectured that every planar graph is in \textsc{Pure-$4$-Dir}. A proof of West's conjecture would have provided an alternate proof for the famous Four Color Theorem~\cite{appel1977} for planar graphs. Indeed, every bipartite planar graph is in $\PureTwoDir$ and every $3$-colourable planar graph is in \textsc{Pure-$3$-Dir}~\cite{hartman1991,gonccalves2019-3-dir}. However, Gon{\c{c}}alves~\cite{gonccalves2020not} proved in 2020 that West's conjecture is false.

Motivated by the long history of research on geometric intersection representations of planar graphs, we show in this paper that finding geometric intersection representations of apex graphs is $\NP$-hard. Our proof technique (\autoref{sec:proof-technique}) deviates considerably from earlier $\NP$-hardness proofs for recognition of geometric intersection graphs, all of which were based on different variants of 3-SAT. We elaborate on this below.

In a highly influential paper, Kratochv{\'\i}l~\cite{kratochvil1994} introduced a variant of 3-SAT called \textsc{4-Bounded Planar 3-Connected 3-SAT (4P3C3SAT)}, and showed that it is $\NP$-hard. He then used 4P3C3SAT to show that recognizing $\PureTwoDir$ is $\NP$-hard. Chmel~\cite{chmel2020} reduced the \textsc{4P3C3SAT} problem to show that recognizing intersection graphs of \textsc{L}-shapes is $\NP$-hard. Kratochv{\'\i}l \& Pergel~\cite{kratochvil2007geometric} used 4P3C3SAT to show that recognizing intersection graphs of line segments is $\NP$-hard even if the inputs are restricted to graphs with large girth.  Musta{\c{t}}{\u{a}} \& Pergel~\cite{mustac2019} also used 4P3C3SAT to generalize and strengthen the above result by showing that there is no polynomial time recognizable graph class that lies between $\PureTwoDir$ and $\String$, even when the inputs are restricted to graphs with arbitrarily large girth and maximum degree at most $8$. Chaplick, Jel{\'\i}nek, Kratochv{\'\i}l \& Vysko{\v{c}}il~\cite{chaplick2012} used similar techniques to show that the recognition of intersection graphs of rectilinear curves having at most $k$ bends (for every constant $k$) is $\NP$-hard.

The construction of the variable and clause gadgets used in these reductions is sometimes quite involved, which ends up making the proofs rather complicated. Perhaps the reason behind this is that even though the incidence graph (of clauses and variables) of a 4P3C3SAT instance is planar, the variable and clause gadgets are non-planar, producing graphs that are far from planar. We overcome this difficulty by reducing from the $\NP$-hard $\PHPCa$ $\PHPCb$ problem~\cite{auer2011}, which in turn was inspired by the $\NP$-hard $\PHCC$ problem~\cite{wigderson1982}. Reducing from this problem allows us to use the natural notion of planarity, making our proofs easier to follow. We explain this in the next section (\autoref{sec:proof-technique}).

%\begin{comment}

\begin{figure}
	\centering
	\scalebox{1}{
		\begin{tabular}{ccc}
			\rotatebox{90}{
				\begin{tikzpicture}
				
				\draw[densely dotted] (0,0) -- (7,0);
				\def\A{4,0}
				\def\B{3,0}
				\def\C{2,0}
				\def\D{1,0}
				\def\E{0,0}
				\def\F{5,0}
				\def\G{6,0}
				\def\H{7,0}
				\renewcommand{\vertexset}{(a,\A,blue!60!white,,,blue!60!white),(b,\B,blue!60!white,,,blue!60!white),(c,\C,blue!60!white,,,blue!60!white),(d,\D,blue!60!white,,,blue!60!white),(e,\E,blue!60!white,,,blue!60!white),(f,\F,blue!60!white,,,blue!60!white),(g,\G,blue!60!white,,,blue!60!white),(h,\H,blue!60!white,,,blue!60!white)}
				
				\renewcommand{\edgeset}{(c,d,,,-0.5),(d,e,,,-0.5),(c,e,,,-1),(a,d,,,0.5),(f,d,,,0.8), (g,d,,,1), (h,f,,,-0.5),(h,a,,,-0.8)}
				
				\renewcommand{\defradius}{0.07}
				
				\drawgraph
				
				\draw (1,0) -- (1,-1) -- (-1,-1) -- (-1,1) -- (3,1) -- (3,0);
				
				\draw (6,0) -- (6,-1.5) -- (-1.5,-1.5) -- (-1.5,1.5) -- (4,1.5) -- (4,0);
				
				\node[left] at (-0.1,0) {\rotatebox{270} {\large $v_1$}};
				\node[left] at (0.9,-0.2) {\rotatebox{270} {\large $v_2$}};
				\node[right] at (2,0.2) {\rotatebox{270} {\large $v_3$}};
				\node[right] at (3,0.2) {\rotatebox{270} {\large $v_4$}};
				\node[right] at (4.1,0) {\rotatebox{270} {\large $v_5$}};
				\node[right] at (5.1,0) {\rotatebox{270} {\large $v_6$}};
				\node[right] at (6.1,0) {\rotatebox{270} {\large $v_7$}};
				\node[right] at (7.1,0) {\rotatebox{270} {\large $v_8$}};
				\end{tikzpicture}} & \rotatebox{90}{\begin{tikzpicture}
				
				\draw[densely dotted] (0,0) -- (7,0);
				
				\def\A{4,0}
				\def\B{3,0}
				\def\C{2,0}
				\def\D{1,0}
				\def\E{0,0}
				\def\F{5,0}
				\def\G{6,0}
				\def\H{7,0}
				
				\renewcommand{\vertexset}{(a,\A,blue!60!white,,,blue!60!white),(b,\B,blue!60!white,,,blue!60!white),(c,\C,blue!60!white,,,blue!60!white),(d,\D,blue!60!white,,,blue!60!white),(e,\E,blue!60!white,,,blue!60!white),(f,\F,blue!60!white,,,blue!60!white),(g,\G,blue!60!white,,,blue!60!white),(h,\H,blue!60!white,,,blue!60!white)}
				
				\renewcommand{\edgeset}{(c,d,,,-0.5),(d,e,,,-0.5),(c,e,,,-1),(a,d,,,0.5),(f,d,,,0.8), (g,d,,,1), (h,f,,,-0.5),(h,a,,,-0.8)}
				
				\renewcommand{\defradius}{0.07}
				
				\drawgraph
				
				\draw (1,0) -- (1,-1) -- (-1,-1) -- (-1,1) -- (3,1) -- (3,0);
				
				\draw (6,0) -- (6,-1.5) -- (-1.5,-1.5) -- (-1.5,1.5) -- (4,1.5) -- (4,0);
				
				\filldraw[magenta!70!black] (0.825,0.25) circle (0.05cm);
				\filldraw[magenta!70!black] (0.625,0.37) circle (0.05cm);
				\filldraw[magenta!70!black] (0.4,0.37) circle (0.05cm);
				
				\filldraw[magenta!70!black] (1.15,0.25) circle (0.05cm);
				\filldraw[magenta!70!black] (1.35,0.37) circle (0.05cm);
				\filldraw[magenta!70!black] (1.53,0.4) circle (0.05cm);
				
				\filldraw[magenta!70!black] (0.625,0.7) circle (0.05cm);
				\filldraw[magenta!70!black] (1.35,0.7) circle (0.05cm);
				\filldraw[magenta!70!black] (1,0.75) circle (0.05cm);
				
				\filldraw[magenta!70!black] (1,1) circle (0.05cm);
				\filldraw[magenta!70!black] (-1,0) circle (0.05cm);
				\filldraw[magenta!70!black] (0,-1) circle (0.05cm);
				
				\filldraw[magenta!70!black] (1,1.5) circle (0.05cm);
				\filldraw[magenta!70!black] (-1.5,0) circle (0.05cm);
				\filldraw[magenta!70!black] (2,-1.5) circle (0.05cm);
				
				\filldraw[magenta!70!black] (2.75,-0.375) circle (0.05cm);
				\filldraw[magenta!70!black] (3,-0.35) circle (0.05cm);
				\filldraw[magenta!70!black] (3.25,-0.32) circle (0.05cm);
				
				\filldraw[magenta!70!black] (3.45,-0.575) circle (0.05cm);
				\filldraw[magenta!70!black] (3.7,-0.555) circle (0.05cm);
				\filldraw[magenta!70!black] (3.95,-0.49) circle (0.05cm);
				
				\filldraw[magenta!70!black] (4.15,-0.72) circle (0.05cm);
				\filldraw[magenta!70!black] (4.4,-0.67) circle (0.05cm);
				\filldraw[magenta!70!black] (4.65,-0.62) circle (0.05cm);
				
				\filldraw[magenta!70!black] (5.25,0.22) circle (0.05cm);
				\filldraw[magenta!70!black] (5.5,0.3) circle (0.05cm);
				\filldraw[magenta!70!black] (5.75,0.38) circle (0.05cm);
				
				\filldraw[magenta!70!black] (4.75,0.5) circle (0.05cm);
				\filldraw[magenta!70!black] (5,0.575) circle (0.05cm);
				\filldraw[magenta!70!black] (5.25,0.6) circle (0.05cm);

				\node[left] at (-0.1,0) {\rotatebox{270} {\large $v_1$}};
				\node[left] at (0.9,-0.2) {\rotatebox{270} {\large $v_2$}};
				\node[right] at (2,0.2) {\rotatebox{270} {\large $v_3$}};
				\node[right] at (3,0.2) {\rotatebox{270} {\large $v_4$}};
				\node[right] at (4.1,0) {\rotatebox{270} {\large $v_5$}};
				\node[right] at (5.1,0) {\rotatebox{270} {\large $v_6$}};
				\node[right] at (6.1,0) {\rotatebox{270} {\large $v_7$}};
				\node[right] at (7.1,0) {\rotatebox{270} {\large $v_8$}};
				\end{tikzpicture}} & \rotatebox{90}{\begin{tikzpicture}
				
				\draw[densely dotted] (0,0) -- (7,0);
				
				\def\A{4,0}
				\def\B{3,0}
				\def\C{2,0}
				\def\D{1,0}
				\def\E{0,0}
				\def\F{5,0}
				\def\G{6,0}
				\def\H{7,0}
				\def\AP{3,-3.3}
				
				\renewcommand{\vertexset}{(ap,\AP,magenta!70!black,0.2,,magenta!70!black), (a,\A,blue!60!white,,,blue!60!white),(b,\B,blue!60!white,,,blue!60!white),(c,\C,blue!60!white,,,blue!60!white),(d,\D,blue!60!white,,,blue!60!white),(e,\E,blue!60!white,,,blue!60!white),(f,\F,blue!60!white,,,blue!60!white),(g,\G,blue!60!white,,,blue!60!white),(h,\H,blue!60!white,,,blue!60!white)}
				
				\renewcommand{\edgeset}{(ap,h,gray,,-0.5),(ap,g,gray,,-0.5),(ap,f,gray,,-0.5),(ap,a,gray,,-0.5),(ap,b,gray,,-0.3),(ap,c,gray,,0.35),(ap,d,gray,,0.5),(ap,e,gray,,0.5),(c,d,,,-0.5),(d,e,,,-0.5),(c,e,,,-1),(a,d,,,0.5),(f,d,,,0.8), (g,d,,,1), (h,f,,,-0.5),(h,a,,,-0.8)}
				
				\renewcommand{\defradius}{0.07}
				
				\drawgraph
				
				\draw (1,0) -- (1,-1) -- (-1,-1) -- (-1,1) -- (3,1) -- (3,0);
				
				\draw (6,0) -- (6,-1.5) -- (-1.5,-1.5) -- (-1.5,1.5) -- (4,1.5) -- (4,0);
				
				\filldraw[magenta!70!black] (0.825,0.25) circle (0.05cm);
				\filldraw[magenta!70!black] (0.625,0.37) circle (0.05cm);
				\filldraw[magenta!70!black] (0.4,0.37) circle (0.05cm);
				
				\filldraw[magenta!70!black] (1.15,0.25) circle (0.05cm);
				\filldraw[magenta!70!black] (1.35,0.37) circle (0.05cm);
				\filldraw[magenta!70!black] (1.53,0.4) circle (0.05cm);
				
				\filldraw[magenta!70!black] (0.625,0.7) circle (0.05cm);
				\filldraw[magenta!70!black] (1.35,0.7) circle (0.05cm);
				\filldraw[magenta!70!black] (1,0.75) circle (0.05cm);
				
				\filldraw[magenta!70!black] (1,1) circle (0.05cm);
				\filldraw[magenta!70!black] (-1,0) circle (0.05cm);
				\filldraw[magenta!70!black] (0,-1) circle (0.05cm);
				
				\filldraw[magenta!70!black] (1,1.5) circle (0.05cm);
				\filldraw[magenta!70!black] (-1.5,0) circle (0.05cm);
				\filldraw[magenta!70!black] (2,-1.5) circle (0.05cm);
				
				\filldraw[magenta!70!black] (2.75,-0.375) circle (0.05cm);
				\filldraw[magenta!70!black] (3,-0.35) circle (0.05cm);
				\filldraw[magenta!70!black] (3.25,-0.32) circle (0.05cm);
				
				\filldraw[magenta!70!black] (3.45,-0.575) circle (0.05cm);
				\filldraw[magenta!70!black] (3.7,-0.555) circle (0.05cm);
				\filldraw[magenta!70!black] (3.95,-0.49) circle (0.05cm);
				
				\filldraw[magenta!70!black] (4.15,-0.72) circle (0.05cm);
				\filldraw[magenta!70!black] (4.4,-0.67) circle (0.05cm);
				\filldraw[magenta!70!black] (4.65,-0.62) circle (0.05cm);
				
				\filldraw[magenta!70!black] (5.25,0.22) circle (0.05cm);
				\filldraw[magenta!70!black] (5.5,0.3) circle (0.05cm);
				\filldraw[magenta!70!black] (5.75,0.38) circle (0.05cm);
				
				\filldraw[magenta!70!black] (4.75,0.5) circle (0.05cm);
				\filldraw[magenta!70!black] (5,0.575) circle (0.05cm);
				\filldraw[magenta!70!black] (5.25,0.6) circle (0.05cm);

				\node[left] at (-0.1,0) {\rotatebox{270} {\large $v_1$}};
				\node[left] at (0.9,-0.2) {\rotatebox{270} {\large $v_2$}};
				\node[right] at (2,0.2) {\rotatebox{270} {\large $v_3$}};
				\node[right] at (3,0.2) {\rotatebox{270} {\large $v_4$}};
				\node[right] at (4.1,0) {\rotatebox{270} {\large $v_5$}};
				\node[right] at (5.1,0) {\rotatebox{270} {\large $v_6$}};
				\node[right] at (6.1,0) {\rotatebox{270} {\large $v_7$}};
				\node[right] at (7.1,0) {\rotatebox{270} {\large $v_8$}};
				\node at (2.4,-3.3) {\rotatebox{270} {\LARGE $a$}};
				
				\end{tikzpicture}}\\
			(a) & (b) & (c)\\ & & \\ 
			& \begin{tikzpicture}
			\draw[magenta!70!black, ultra thick] (0,7) -- (0,0);
			
			\draw[blue!60!white] (0,6.5) -- (-1.5,6.5);
			\draw[blue!60!white] (0,5.7) -- (3,5.7);
			\draw[blue!60!white] (-0.7,4.9) -- (1,4.9);
			\draw[blue!60!white] (-3,4.1) -- (0.5,4.1);
			\draw[blue!60!white] (-2,3.3) -- (0,3.3);
			\draw[blue!60!white] (-1.4,2.5) -- (0,2.5);
			\draw[blue!60!white] (-0.7,1.7) -- (2,1.7);
			\draw[blue!60!white] (-1.4,0.9) -- (0,0.9);
			
			% Cross over edges
			
			\draw[magenta!70!black] (2.9,5.8) -- (2.9,-1.1);
			\draw[magenta!70!black] (3,-1) -- (-3,-1);
			\draw[magenta!70!black] (-2.9,-1.1) -- (-2.9,4.2);
			
			\draw[magenta!70!black] (1.9,1.8) -- (1.9,-0.6);
			\draw[magenta!70!black] (2,-0.5) -- (-2,-0.5);
			\draw[magenta!70!black] (-1.9,-0.6) -- (-1.9,3.4);
			
			% left edges
			
			\draw[magenta!70!black] (1.7,5.8) -- (1.7,3.8);
			\draw[magenta!70!black] (1.6,3.9) -- (1.6,1.6);
			\draw[magenta!70!black] (1.5,3.85) -- (1.8,3.85);
			
			\draw[magenta!70!black] (0.9,5) -- (0.9,3.2);
			\draw[magenta!70!black] (0.8,3.3) -- (0.8,1.6);
			\draw[magenta!70!black] (0.7,3.25) -- (1,3.25);
			
			\draw[magenta!70!black] (0.4,4.2) -- (0.4,2.5);
			\draw[magenta!70!black] (0.3,2.6) -- (0.3,1.6);
			\draw[magenta!70!black] (0.2,2.55) -- (0.5,2.55);

			% Right edges
			
			\draw[magenta!70!black] (-0.4,6.6) -- (-0.4,5.8);
			\draw[magenta!70!black] (-0.5,5.9) -- (-0.5,4.8);
			\draw[magenta!70!black] (-0.6,5.85) -- (-0.3,5.85);
			
			\draw[magenta!70!black] (-1.3,6.6) -- (-1.3,5.2);
			\draw[magenta!70!black] (-1.4,5.3) -- (-1.4,4);
			\draw[magenta!70!black] (-1.5,5.25) -- (-1.2,5.25);
			
			\draw[magenta!70!black] (-0.4,2.6) -- (-0.4,2);
			\draw[magenta!70!black] (-0.5,2.1) -- (-0.5,1.6);
			\draw[magenta!70!black] (-0.6,2.05) -- (-0.3,2.05);
			
			\draw[magenta!70!black] (-0.4,1.8) -- (-0.4,1.3);
			\draw[magenta!70!black] (-0.5,1.4) -- (-0.5,0.8);
			\draw[magenta!70!black] (-0.6,1.35) -- (-0.3,1.35);
			
			\draw[magenta!70!black] (-1.2,2.6) -- (-1.2,1.6);
			\draw[magenta!70!black] (-1.3,1.7) -- (-1.3,0.8);
			\draw[magenta!70!black] (-1.1,1.65) -- (-1.4,1.65);
			
			\node[above,blue!60!white] at (-1,6.5) {\large $\curve{v_8}$};
			\node[above, blue!60!white] at (2,5.7) {\large $\curve{v_7}$};
			\node[above,blue!60!white] at (1,4.9) {\large $\curve{v_6}$};
			\node[above, blue!60!white] at (-2.5,4.1) {\large $\curve{v_5}$};
			\node[above, blue!60!white] at (-1.5,3.3) {\large $\curve{v_4}$};
			\node[above, blue!60!white] at (-1,2.5) {\large $\curve{v_3}$};
			\node[below, blue!60!white] at (1,1.7) {\large $\curve{v_2}$};
			\node[below, blue!60!white] at (-1,0.95) {\large $\curve{v_1}$};
			
			\node[above,magenta!70!black] at (0,7) {\LARGE $\curve{a}$};
			\end{tikzpicture} & \\
			& (d) & 
	\end{tabular}}
	\caption{(a) A front line drawing (\autoref{def:frontlinedrawing}) of a planar graph $G$, where $G$ is a yes-instance of $\PHPCa$ $\PHPCb$ (\autoref{def:planhampathcomp}), and the vertical dotted line denotes the front line $\ellfront$; (b) $\subdivide{G}$ for $k=3$; (c) $\apexgraph{G}$; (d) $C$, a $\PureTwoDir$ representation of $\apexgraph{G}$ (\autoref{subsec:4apex2dir}).}\label{fig:representation}
\end{figure}

%\end{comment}
\section{Proof Techniques}\label{sec:proof-technique}

Let us now describe the $\PHPC$ problem~\cite{auer2011}. A Hamiltonian path in a graph is a path that visits each vertex of the graph exactly once.

\begin{Definition}\label{def:planhampathcomp}
	$\PHPC$ is the following decision problem.
	%\begin{description}
	
	\noindent \textbf{Input:} A planar graph $G$.
	
	\noindent \textbf{Output:} Yes, if $G$ is a subgraph of a planar graph with a Hamiltonian path; no, otherwise.
	%\end{description}
\end{Definition}

Auer \& Glei{\ss}ner~\cite{auer2011} showed that this problem (\autoref{def:planhampathcomp}) is $\NP$-hard. For their proof, they defined the following constrained planar embedding (\autoref{fig:representation} (a)).%that holds only for a subclass of planar graphs.

\begin{Definition}[Front line drawing] \label{def:frontlinedrawing} A front line drawing of a planar graph $G$ is a planar embedding of $G$ such that its vertices lie on a vertical line segment $\ellfront$ (called the~\emph{front line} by Auer \& Glei{\ss}ner~\cite{auer2011}) and its edges do not cross each other or $\ellfront$. Let $e=v_iv_j$ be an edge of $G$ such that $e$ connects to the vertex $v_i$ from the left of $\ellfront$, and to the vertex $v_j$ from the right of $\ellfront$. We call such an edge a~\emph{crossover edge}. \autoref{fig:representation} (a) is a front line drawing with two crossover edges, namely $v_5v_7$ and $v_4v_2$.
\end{Definition}

\begin{Observation}[Auer \& Glei{\ss}ner~\cite{auer2011}] \label{obs:frontlinedrawing} A planar graph $G$ is a yes-instance of the $\PHPCa$ $\PHPCb$ problem if and only if $G$ admits a front line drawing.
\end{Observation}

Using this observation, they showed that deciding whether a graph $G$ admits a front line drawing is $\NP$-hard, implying the required theorem.

\begin{Theorem}[Auer \& Glei{\ss}ner~\cite{auer2011}] \label{thm:auer}
	$\PHPC$ is $\NP$-hard.
\end{Theorem}

We will use~\autoref{thm:auer} to show our main result (\autoref{thm:main}). We show $\NP$-hardness for graph classes ``sandwiched'' between two classes of geometric intersection graphs, similar to Musta{\c{t}}{\u{a}} \& Pergel~\cite{mustac2019}. A more technical formulation of our main result is as follows.

% Note that this formulation allows us to prove $\NP$-hardness for all graph classes that lie between $\PureTwoDir$ and $\OneString\cup \Conv$.

\begin{Theorem}\label{thm:mainrestat}
	For every planar graph $G$ and positive integer $g$, there exists a bipartite apex graph $\apexgraph{G}$ of girth at least $g$ which can be obtained in polynomial time from $G$, satisfying the following properties.
	\begin{enumerate}%[label=\alph*]
		\item[(a)] If $\apexgraph{G}$ is in $\OneString$, then $G$ is a yes-instance of $\PHPCa$ $\PHPCb$.
		\item[(b)] If $G$ is a yes-instance of $\PHPC$, then $\apexgraph{G}$ is in $\PureTwoDir$. 
	\end{enumerate}
\end{Theorem}

%It is easy to see that~\autoref{thm:mainrestat} implies our main result (\autoref{thm:main}).

\begin{proof}[Proof of~\autoref{thm:main} assuming~\autoref{thm:mainrestat}] Let $\cG$ be a graph class satisfying the condition $\PureTwoDir\subseteq\cG\subseteq\OneString$, and let $G$ be a planar graph. If $G$ is a yes-instance of $\PHPCa$ $\PHPCb$, then by~\autoref{thm:mainrestat} (b), $\apexgraph{G}\in\PureTwoDir\subseteq\cG$. And if $\apexgraph{G}\in\cG\subseteq\OneString$, then by~\autoref{thm:mainrestat} (a), $G$ is a yes-instance of $\PHPCa$ $\PHPCb$.
	
Thus, $\apexgraph{G}\in\cG$ if and only if $G$ is a yes-instance of $\PHPCa$ $\PHPCb$. Since $\PHPCa$ $\PHPCb$ is $\NP$-hard (\autoref{thm:auer}) and $\apexgraph{G}$ can be obtained in polynomial time from $G$, this implies that deciding whether the bipartite apex graph $\apexgraph{G}$ belongs to $\cG$ is $\NP$-hard.
\end{proof}

Therefore, as~\autoref{thm:mainrestat} implies our main result (\autoref{thm:main}), the rest of this paper is devoted to the proof of~\autoref{thm:mainrestat}.

\section{Proof of the Main Result}\label{sec:proof}

%\begin{comment}

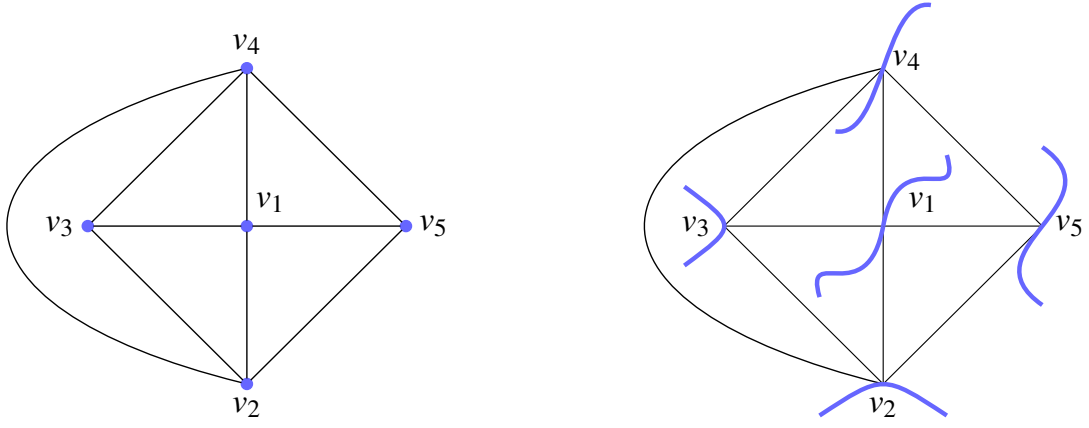
\begin{figure}[ht]
	\vspace{-1.5em}
	\begin{centering}
		\resizebox{1\textwidth}{!}{
			\begin{tikzpicture}
			
			\def\A{4,2}
			\def\B{2,0}
			\def\C{2,4}
			\def\D{2,2}
			\def\E{0,2}
			\renewcommand{\vertexset}{(a,\A,blue!60!white,,,blue!60!white),(b,\B,blue!60!white,,,blue!60!white),(c,\C,blue!60!white,,,blue!60!white),(d,\D,blue!60!white,,,blue!60!white),(e,\E,blue!60!white,,,blue!60!white)}
			\renewcommand{\edgeset}{
				(a,b),(a,d),(a,c),
				(b,c,,,4),(b,d),(b,e),
				(c,d),(c,e),
				(d,e)}
			\renewcommand{\defradius}{0.07}
			\drawgraph
			\node at (4.35,2) {\large $v_5$};
			\node at (2,-0.3) {\large $v_2$};
			\node at (2,4.3) {\large $v_4$};
			\node at (2.3,2.3) {\large $v_1$};
			\node at (-0.35,2) {\large $v_3$};
			
			%%%%%%%%%%%%%%%%%%%%%%%%%%%%%%%%%%%%%%%%%%%%%%%%%%%%%%%
			
			%\def\A{12,2}
			\def\B{10,0}
			\def\C{10,4}
			\renewcommand{\vertexset}{(b,\B,white,,,white),(c,\C,white,,,white)}
			\renewcommand{\edgeset}{(b,c,,,4)}
			\renewcommand{\defradius}{0.01}
			\drawgraph
			
			\draw (12,2) -- (10,4) -- (8,2) -- (12,2) -- (10,0) -- (8,2);
			\draw (10,0) -- (10,4);
			
			\draw [ultra thick, blue!60!white] (12,1) .. controls (11,1.7) and (13,2.3) .. (12,3); %v5
			\draw [ultra thick, blue!60!white] (9.4,3.2) .. controls (10,3.1) and (10,4.9) .. (10.6,4.8); %v4
			\draw [ultra thick, blue!60!white] (10.8,2.9) .. controls (11,2.2) and (10.2,3.1) .. (10,2); %v1
			\draw [ultra thick, blue!60!white] (10,2) .. controls (9.8,0.9) and (9,1.8) .. (9.2,1.1); %v1
			\draw [ultra thick, blue!60!white] (9.2,-0.4) .. controls (10,0.13) .. (10.8,-0.4); %v2
			\draw [ultra thick, blue!60!white] (7.5,2.5) .. controls (8.17,2) .. (7.5,1.5); %v3
			
			\node at (12.35,2) {\large $v_5$};
			\node at (10,-0.3) {\large $v_2$};
			\node at (10.3,4.1) {\large $v_4$};
			\node at (10.5,2.3) {\large $v_1$};
			\node at (7.65,2) {\large $v_3$};
			
			\end{tikzpicture}
		}
		\caption{(Left) A standard representation of a planar graph, where the vertices are points and the edges are strings. (Right) A planarizable representation (\autoref{def:planar}) of the same graph, where the vertices as well as the edges are strings.} \label{fig:planargraphdeffig}
	\end{centering}
\end{figure}

%\end{comment}

\subsection{Construction of the Apex Graph}\label{sec:construct}

We begin our proof of~\autoref{thm:mainrestat} by describing the construction of $\apexgraph{G}$. Let $G$ be a planar graph. $\apexgraph{G}$ is constructed in two steps. 
\begin{equation*}
G\rightarrow \subdivide{G}\rightarrow \apexgraph{G}.
\end{equation*}
Let $g\geq 6$ be a positive integer, and $k\geq 3$ be the minimum odd integer greater or equal to $g-3$. Let $\subdivide{G}$ be the full $k$-subdivision of $G$, \ie, $\subdivide{G}$ is the graph obtained by replacing each edge of $G$ by a path with $k+1$ edges. \autoref{fig:representation} (a) denotes a graph $G$, and~\autoref{fig:representation} (b) denotes the full $3$-subdivision of $G$. Formally, we replace each $e=(x,y)\in E(G)$ by the path $(x, u^1_e, u^2_e, u^3_e, \ldots, u^k_e, y)$.
\begin{align*}
V(\subdivide{G}) &= V(G) \cup \{u^1_e, u^2_e, u^3_e, \ldots, u^k_e\mid e\in E(G)\}; \\
E(\subdivide{G}) &= \{x u^1_e, u^1_e u^2_e, u^2_e u^3_e,\ldots, u^{k-1}_e u^k_e, u^k_e y \mid e=xy\in E(G)\}.
\end{align*}
We call the vertices of $V(G)\subseteq V(\subdivide{G})$ as the~\emph{original vertices} of $\subdivide{G}$ and the remaining vertices as the ~\emph{subdivision vertices} of $\subdivide{G}$.
%\begin{Observation}
%If $G$ is a planar graph, then $\subdivide{G}$ is a bipartite planar graph.
%\end{Observation}
Finally, we construct $\apexgraph{G}$ by adding a new vertex $a$ to $\subdivide{G}$ and making it adjacent to all the original vertices of $\subdivide{G}$ (\autoref{fig:representation} (c)). Formally, $\apexgraph{G}$ is defined as follows.
\begin{align*}
V(\apexgraph{G}) &= V(\subdivide{G}) \cup \{a\}; \\
E(\apexgraph{G}) &= E(\subdivide{G}) \cup \{a v \mid v\in V(G)\}.
\end{align*}

\begin{Observation}
	If $G$ is a planar graph, then $\apexgraph{G}$ is a bipartite apex graph of girth at least $g$.
\end{Observation}
\begin{proof} $G$ is a planar graph and subdivision does not affect planarity, so $\subdivide{G}$ is also planar, implying that $\apexgraph{G}$ is an apex graph. The vertex set of $\apexgraph{G}$ can be expressed as the disjoint union of two sets $A$ and $B$, where
	\begin{align*}
	A &= \{x\mid x\in V(G) \}\cup\{u^i_e\mid e\in E(G), i \text{ is even}\};\\
	B &= \{a\} \cup \{u^i_e\mid e\in E(G), i \text{ is odd}\}.
	\end{align*}
	Note that $A$ induces an independent set in $\apexgraph{G}$, and so does $B$. Thus, $\apexgraph{G}$ is a bipartite apex graph. As for the girth, note that every cycle in $\apexgraph{G}$ contains at least $k+2$ vertices $x, u^1_e, u^2_e, u^3_e, \ldots, u^k_e, y$, for some $e=(x,y)\in E(G)$. At least one more vertex is needed to complete the cycle, implying that the girth of $\apexgraph{G}$ is at least $k+3\geq g$.
\end{proof}

It is easy to see that this entire construction of $\apexgraph{G}$ from $G$ can be carried out in polynomial time. In~\autoref{subsec:1string4apex} and~\autoref{subsec:4apex2dir}, we will prove~\autoref{thm:mainrestat} (a) and~\autoref{thm:mainrestat} (b), respectively.

\subsection{Proof of Theorem~\ref{thm:mainrestat} (a)}\label{subsec:1string4apex}

In this section, we will show that if $\apexgraph{G}$ is in $\OneString$, then $G$ is a yes-instance of $\PHPCa$ $\PHPCb$. In other words, if $\apexgraph{G}$ has a $\OneString$ representation, then $G$ is a subgraph of a planar graph with a Hamiltonian path.

%Let $\mathcal{S}$ be a set of simple curves in the plane such that every pair of intersecting curves in $\mathcal{S}$ cross at most once. Let $G$ be the intersection graph of $\mathcal{S}$. Thus, $\mathcal{S}$ is a $\OneString$ representation of $G$. A $\PureTwoDir$  representation of a graph can be defined analogously.

In our proofs, we will demonstrate the planarity of our graphs by embedding them in the plane. Typically, a planar graph is defined as a graph whose vertices are~\emph{points} in the plane and edges are~\emph{strings} connecting pairs of points such that no two strings intersect (except possibly at their end points). The same definition holds in more generality, \ie, if the vertices are also allowed to be strings (\autoref{fig:planargraphdeffig}). Let us state this formally.

\begin{Definition} [Planarizable representation of a graph] \label{def:planar}
	A graph $G$ on $n$ vertices and $m$ edges is said to admit a planarizable representation if there are two mutually disjoint sets of strings $V$ and $E\ ($with $|V|=n$ and $|E|=m)$ in the plane such that
	\begin{itemize}
		\item the strings of $V$ correspond to the vertices of $G$ and the strings of $E$ correspond to the edges of $G$;
		\item no two strings of $V$ intersect;
		\item no two strings of $E$ intersect, except possibly at their end points;
		\item apart from its two end points, a string of $E$ does not intersect any string of $V$;
		\item for every vertex $v$ and every edge $e=(x,y)$ of $G$, an end point of the string corresponding to $e$ intersects the string corresponding to $v$ if and only if $v=x$ or $v=y$.
	\end{itemize}
	\autoref{fig:planargraphdeffig} illustrates a planar graph and a planarizable representation of it.
\end{Definition}

\begin{Lemma} \label{lem:regstringrepplanar}
	A graph admits a planarizable representation if and only if it is planar.
\end{Lemma}

%\begin{Lemma} %\label{lem:convexrepplanar}
%A graph admits a convex representation if and only if it is planar.
%\end{Lemma}

\autoref{lem:regstringrepplanar} may seem obvious. For completeness, we provide a formal proof of it in~\autoref{sec:planarrep}. We now use this lemma to prove~\autoref{thm:mainrestat} (a).

\begin{proof}[Proof of~\autoref{thm:mainrestat} (a)] Given $\apexgraph{G}\in\OneString$, we will show that the planar graph $G$ is a yes-instance of $\PHPCa$ $\PHPCb$. Let $C$ be a $\OneString$ representation of $\apexgraph{G}$ in the plane. It is helpful to follow~\autoref{fig:phpc1string} while reading this proof. We will use $C$ to construct a graph $\planar{G}$ with the following properties.
	\begin{enumerate}
		\item[(a)] $\planar{G}$ is a supergraph of $G$ on the same vertex set as $G$.
		\item[(b)] $\planar{G}$ is planar.
		\item[(c)] $\planar{G}$ has a Hamiltonian path.
	\end{enumerate}
	Note that (a), (b), (c) together imply that $G$ is a subgraph of a planar graph with a Hamiltonian path (\ie, $G$ is a yes-instance of $\PHPC$). Let $n=|V(G)|$ and assume that $n\geq 4$. Along with our construction of $\planar{G}$, we will also describe $\draw(\planar{G})$, a planarizable representation (\autoref{def:planar}) of $\planar{G}$ in the plane.
	
	In $C$, consider the strings corresponding to the $n$ original vertices (the ones in blue in~\autoref{fig:phpc1string}) of $G$. Since the original vertices form an independent set in $\apexgraph{G}$, the blue strings are pairwise disjoint. We add these $n$ strings to $\draw(\planar{G})$, which correspond to the $n$ vertices of $\planar{G}$.
	
	%\begin{comment}
	
	\begin{figure}
		\vspace{-4em}
		\begin{centering}
			\resizebox{1\textwidth}{!}{
				\begin{tikzpicture}
				
				\def\A{-4.5,0}
				\def\B{-6.5,-2}
				\def\C{-6.5,2}
				\def\D{-6.5,0}
				\def\E{-8.5,0}
				
				\renewcommand{\vertexset}{(a,\A,blue!60!white,,,blue!60!white),(b,\B,blue!60!white,,,blue!60!white),(c,\C,blue!60!white,,,blue!60!white),(d,\D,blue!60!white,,,blue!60!white),(e,\E,blue!60!white,,,blue!60!white)}
				\renewcommand{\edgeset}{
					(a,b),(a,d),(a,c),
					(b,c,,,4),(b,d),(b,e),
					(c,d),(c,e),
					(d,e)}
				\renewcommand{\defradius}{0.07}
				\drawgraph
				
				\filldraw[magenta!70!black] (-6,0) circle (1.5pt);
				\filldraw[magenta!70!black] (-5.5,0) circle (1.5pt);
				\filldraw[magenta!70!black] (-5,0) circle (1.5pt);
				
				\filldraw[magenta!70!black] (-8,0) circle (1.5pt);
				\filldraw[magenta!70!black] (-7.5,0) circle (1.5pt);
				\filldraw[magenta!70!black] (-7,0) circle (1.5pt);
				
				\filldraw[magenta!70!black] (-6,1.5) circle (1.5pt);
				\filldraw[magenta!70!black] (-5.5,1) circle (1.5pt);
				\filldraw[magenta!70!black] (-5,0.5) circle (1.5pt);
				
				\filldraw[magenta!70!black] (-8,0.5) circle (1.5pt);
				\filldraw[magenta!70!black] (-7.5,1) circle (1.5pt);
				\filldraw[magenta!70!black] (-7,1.5) circle (1.5pt);
				
				\filldraw[magenta!70!black] (-6,-1.5) circle (1.5pt);
				\filldraw[magenta!70!black] (-5.5,-1) circle (1.5pt);
				\filldraw[magenta!70!black] (-5,-0.5) circle (1.5pt);
				
				\filldraw[magenta!70!black] (-8,-0.5) circle (1.5pt);
				\filldraw[magenta!70!black] (-7.5,-1) circle (1.5pt);
				\filldraw[magenta!70!black] (-7,-1.5) circle (1.5pt);
				
				\filldraw[magenta!70!black] (-6.5,0.5) circle (1.5pt);
				\filldraw[magenta!70!black] (-6.5,1) circle (1.5pt);
				\filldraw[magenta!70!black] (-6.5,1.5) circle (1.5pt);
				
				\filldraw[magenta!70!black] (-6.5,-0.5) circle (1.5pt);
				\filldraw[magenta!70!black] (-6.5,-1) circle (1.5pt);
				\filldraw[magenta!70!black] (-6.5,-1.5) circle (1.5pt);
				
				\filldraw[magenta!70!black] (-8,-1.5) circle (1.5pt);
				\filldraw[magenta!70!black] (-9.5,0) circle (1.5pt);
				\filldraw[magenta!70!black] (-8,1.5) circle (1.5pt);
				
				\node at (-4.15,0) {\large $v_5$};
				\node at (-6.5,-2.3) {\large $v_2$};
				\node at (-6.5,2.3) {\large $v_4$};
				\node at (-6.2,0.3) {\large $v_1$};
				\node at (-8.85,0) {\large $v_3$};
				
				%%%%%%%%%%%%%%%%%%%%%%%%%%%%%%%%%%%%%%%%%%%%%%%%%%%%
				
				\draw (4,0)[magenta!70!black, ultra thick] circle [x radius=4cm, y radius=2cm];
				\draw (4,0)[magenta!70!black, ultra thick] circle [x radius=4cm, y radius=2cm];
				
				\node[magenta!70!black] at (8.1,-1.1) {\LARGE $\curve{a}$};
				
				\node[blue!90!white] at (0.7,-2.4) {\large $\curve{v_4}$};
				
				\node[blue!90!white] at (2.5,2.5) {\large $\curve{v_5}$}; 
				
				\node[blue!90!white] at (7,2.5) {\large $\curve{v_1}$};
				
				\node[blue!90!white] at (8.7,1.4) {\large $\curve{v_2}$};
				
				\node[blue!90!white] at (6.4,-2.6) {\large $\curve{v_3}$};

				% 	\draw[magenta!70!black, very thin] (0.85,0) .. controls (1,-1) and (2,-0.5) .. (2.8,-0.2);
				\draw[magenta!70!black, very thin] (0.85,0) -- (1.5,-0.5);\draw[magenta!70!black, very thin] (1.2,-0.4) -- (2.3,-0.4); \draw[magenta!70!black, very thin] (1.9,-0.5) -- (2.9,0);
				%e14
				
				% 	\draw[magenta!70!black, very thin] (1.7,1) .. controls (2.4,1.8) and (3,0.8) .. (3.2,0.6); 
				\draw[magenta!70!black, very thin] (1.7,1) -- (2.3,1.4); \draw[magenta!70!black, very thin]  (2,1.3) -- (2.8,1.3); \draw[magenta!70!black, very thin]  (2.6,1.4) -- (3.2,0.6);   %e15
				
				\draw[magenta!70!black, very thin] (4.5,-1) -- (5,-0.5);\draw[magenta!70!black, very thin] (4.5,-0.6) -- (5,-0.6); \draw[magenta!70!black, very thin] (4.6,-0.7) -- (4.8,-0.1); %e13
				
				\draw[magenta!70!black, very thin] (5.3,0.8) -- (6.8,0); \draw[magenta!70!black, very thin] (5.9,0.6) -- (5,0.7); \draw[magenta!70!black, very thin] (5.2,0.5) -- (4.3,1.8); %e12
				
				\draw[magenta!70!black, very thin] (-0.5,-0.1) .. controls (-3,4) and (9,6) ..(9.1,2.3); \draw[magenta!70!black, very thin] (9.2,2.7) -- (8.7,1.8); \draw[magenta!70!black, very thin] (-0.4,0.2) -- (-0.85,-0.5); %e24
				
				\draw[magenta!70!black, very thin] (8.1,0.5) .. controls (11,-3) and (9,-2.5) ..(7,-1.9); \draw[magenta!70!black, very thin] (8.2,1.6) -- (8.3,0); \draw[magenta!70!black, very thin] (7.5,-1.9) -- (5.8,-2.3); %e23
				
				\draw[magenta!70!black, very thin] (1.5,-2.3) -- (3.35,-2.6); 	\draw[magenta!70!black, very thin] (2.5,-2.6) -- (4.7,-2.2); 	\draw[magenta!70!black, very thin] (3.7,-2.2) -- (5.8,-2.5); %e34
				
				\draw[magenta!70!black, very thin] (-0.1,0.5) .. controls (-1,2) and (2,3) .. (2.5,3.2);  \draw[magenta!70!black, very thin]  (-0.1, 0.8) -- (-0.4,-0.3); \draw[magenta!70!black, very thin]  (2,3.2) -- (3.1,2.8); %e45
				
				\draw[magenta!70!black, very thin] (3.15,2.6) .. controls (3.8,2) and (6.5,4.8) .. (7.8,3);
				
				\draw[magenta!70!black, very thin] (7.6,3.2) .. controls (7.5,3.2) and (8.8,3) .. (8.2,2.5);
				\draw[magenta!70!black, very thin] (8.25,2.7) --(8.4,1.7);
				%e25

				\draw[blue!60!white, ultra thick] (1,0) .. controls (-5,-1) and (5,-2) .. (0,-3); %v1
				\draw[blue!60!white, ultra thick] (2,1) .. controls (0,1) and (4,2) .. (3,3); %v2
				\draw[blue!60!white, ultra thick] (4,-1) .. controls (8,-1) and (5,-2) .. (6,-3); %v5
				\draw[blue!60!white, ultra thick] (5,0) .. controls (-1,-3) and (7,6) .. (8,2); %v3
				\draw[blue!60!white, ultra thick] (6,-1) .. controls (8,-1) and (5,0) .. (9,2); %v4
				
				\draw[white, fill=white] (3,1.9) rectangle (4.2,2.1);
				
				\filldraw (4.55,1.98) circle (1.5pt) node[anchor=north west] {$p_1$};
				\filldraw (7.45,1.01) circle (1.5pt) node[anchor=east] {$p_2$\text{ }};
				\filldraw (6.1,-1.7) circle (1.5pt) node[anchor=south east] {$p_3$};
				\filldraw (0.03,-0.19) circle (1.5pt) node[anchor=south west] {$p_4$};
				\filldraw (2.48,1.85) circle (1.5pt) node[anchor=north west] {$p_5$};
				
				\end{tikzpicture}
			}
			\caption{(Left) $\subdivide{G}$ ($k=3$) for a planar graph $G$. (Right) $C$, a planarizable representation of $\apexgraph{G}$. The bold red string $\curve{a}$ denotes the apex vertex of $\apexgraph{G}$. The blue strings denote the original vertices of $G$, and the thin red strings are $\curve{u^1_e}$, $\curve{u^2_e}$ and $\curve{u^3_e}$, corresponding to the subdivision vertices of $G$.}\label{fig:phpc1string}
		\end{centering}
	\end{figure}
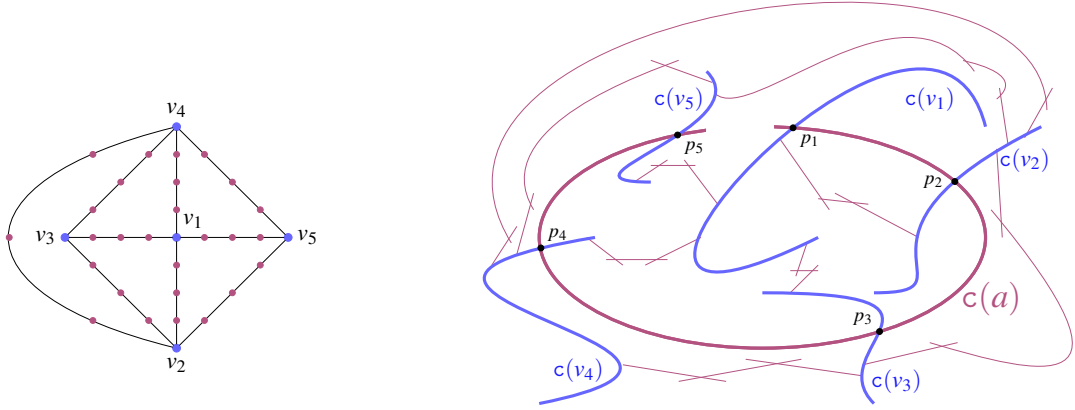
	
	%\end{comment}
	
	\medskip\noindent\textbf{Proof of (c):} So far, $\planar{G}$ has no edge. We will now add $n-1$ edges to $\planar{G}$ to connect these vertices via a Hamiltonian path. Recall that each of the $n$ original vertices is adjacent to the apex vertex $a$ in $\apexgraph{G}$, which means that each of the $n$ blue strings intersects $\curve{a}$ at precisely one point (since $C$ is a $\OneString$ representation). Starting from one end point of $\curve{a}$ and travelling along the curve $\curve{a}$ until we reach its other end point, we encounter these $n$ points one-by-one. Let $(v_1,v_2,\ldots,v_n)$ be the order in which they are encountered.
	
	For each $i\in[n]$, let $p_i$ be the point at which $\curve{v_i}$ intersects $\curve{a}$. For each $i\in[n-1]$, let $\mathtt{s_i}$ be the substring of $\curve{a}$ between $p_i$ and $p_{i+1}$. Add the strings $\mathtt{s_1}, \mathtt{s_2}, \ldots, \mathtt{s_{n-1}}$ as edges to $\draw(\planar{G})$, where $\mathtt{s_i}$ represents the edge between $v_i$ and $v_{i+1}$. Thus the edges corresponding to the $n-1$ strings $\mathtt{s_1}, \mathtt{s_2}, \ldots, \mathtt{s_{n-1}}$ constitute a Hamiltonian path $(v_1,v_2,\ldots,v_n)$ in $\planar{G}$. This shows (c).
	
	\medskip\noindent\textbf{Proof of (a):} To show (a), we need to add all the edges of $G$ to $\planar{G}$ (other than those already added by the previous step), so that $\planar{G}$ becomes a supergraph of $G$. For each edge $e=v_iv_j\in E(G)$, there are $k$ strings $\curve{u^1_e}, \curve{u^2_e}, \ldots, \curve{u^k_e}$ (corresponding to the subdivision vertices $u^1_e, u^2_e, \ldots, u^k_e$ in $\apexgraph{G}$) in $C$. Note that for each $t\in\{1,2,\ldots,k\}$, the string $\curve{u^t_e}$ intersects exactly two other strings. Let $\mathtt{s}(u^t_e)$ be the substring of $\curve{u^t_e}$ between those two intersection points. Let $\mathtt{s_e}$ be the string obtained by concatenating the $k$ substrings thus obtained.
	\begin{equation} \label{eq:concat}
	\mathtt{s_e}\triangleq \bigcup\limits_{t=1}^{k} \mathtt{s}(u^t_e).
	\end{equation}
	If the edge $e=v_iv_j$ is not already present in $\planar{G}$, then add the string $\mathtt{s_e}$ to $\draw(\planar{G})$, where $\mathtt{s_e}$ represents the edge between $v_i$ and $v_j$ (one end point of $\mathtt{s_e}$ lies on $\curve{v_i}$ and the other on $\curve{v_j}$). This completes the construction of $\draw(\planar{G})$, and shows (a).
	
	\medskip\noindent\textbf{Proof of (b):} To show (b), it is enough to show that $\draw(\planar{G})$ is a planarizable representation of $\planar{G}$ (\autoref{lem:regstringrepplanar}). Note that there are three types of strings in $\draw(\planar{G})$: (i) substrings of $\curve{a}$, (ii) strings of the type $\mathtt{s_e}$, for some $e=v_iv_j\in E(G)$, and (iii) $n$ strings corresponding to the original vertices of $G$.
	
	Two strings of type (i) are either disjoint or intersect at their end points, since $\curve{a}$ is non-self-intersecting. More precisely, for each $i\in[n-1]$, the point $p_{i+1}$ (the unique intersection point of $\mathtt{s_i}$ and $\mathtt{s_{i+1}}$) lies on $\curve{v_{i+1}}$, which denotes a vertex in $\draw(\planar{G})$. A string of type (ii) intersects exactly two strings, $\curve{v_i}$ and $\curve{v_j}$, which denote vertices in $\draw(\planar{G})$. Finally, strings of type (iii) are mutually disjoint. This shows (b).
\end{proof}

%\begin{Lemma}\label{lem:super-1-string}
%If $\fourapex{G}$ has a 1-string representation, then $\apexgraph{G}$ has a planarizable representation.
%\end{Lemma}

\subsection{Proof of Theorem~\ref{thm:mainrestat} (b)} \label{subsec:4apex2dir}

In this section, we will show that if $G$ is a yes-instance of $\PHPCa$ $\PHPCb$, then $\apexgraph{G}$ is in $\PureTwoDir$. In other words, if $G$ is a subgraph of a planar graph with a Hamiltonian path, then $\apexgraph{G}$ has a $\PureTwoDir$ representation.

%Due to space constraints, we elucidate the main idea behind our proof here. The full technical details of the proof can be found in~\autoref{sec:proofpuretwodir}.

First we elucidate the main idea behind our proof. Given a front line drawing (\autoref{fig:representation} (a)) of $G$, we will modify it to obtain a $\PureTwoDir$ representation (\autoref{fig:representation} (d)) of $\apexgraph{G}$. The apex segment $\curve{a}$ takes the place of the front line $\ellfront$. The edges of $G$, which were strings in the front line drawing, are replaced by rectilinear piecewise linear curves. If we were allowed a large number of rectilinear pieces for each edge, then this construction is trivial, since every curve can be viewed as a series of infinitesimally small vertical and horizontal segments. Our proof formally justifies that this can always be done when the number of allowed rectilinear pieces is a fixed odd integer greater than or equal to three.

\begin{proof}[Proof of~\autoref{thm:mainrestat} (b)] Given $G$, a yes-instance of $\PHPCa$ $\PHPCb$, we will construct $C$, a $\PureTwoDir$ representation of $\apexgraph{G}$. Recall that the construction of $\apexgraph{G}$ from $G$ uses an intermediate graph $\subdivide{G}$, where $k$ is an odd integer greater than or equal to three. $$G\rightarrow \subdivide{G}\rightarrow \apexgraph{G}.$$

Our proof is by induction on $k$. The major portion of this proof is for $k=3$. At the end, we will show that if the proof works for $k$, then it also works for $k+2$, and consequently for all odd integers $k\geq 3$.

\paragraph{Base case $(k=3)$:} Since $G$ is a subgraph of a planar graph with a Hamiltonian path (\autoref{def:planhampathcomp}), $G$ admits a front line drawing (\autoref{fig:representation} (a)), as per~\autoref{obs:frontlinedrawing}. Our construction of $C$ (\autoref{fig:representation} (d)) is based on and visually motivated by this drawing.

Let $n=|V(G)|$ and assume that $n\geq 4$. Let $(v_1,v_2,\ldots,v_n)$ be the ordering of the vertices on the front line $\ellfront$ from bottom to top. Let $\Eleftedge$ be the set of edges that lie entirely to the left of $\ellfront$, $\Erightedge$ be the set of edges that lie entirely to the right of $\ellfront$, and $\Ecrossedge$ be the set of crossover edges (\autoref{def:frontlinedrawing}). Further, $\Ecrossedge$ is the disjoint union of $\Eaboveedge$ and $\Ebelowedge$, where $\Eaboveedge$ is the set of crossover edges that go above $\ellfront$, and $\Ebelowedge$ is the set of crossover edges that go below $\ellfront$. For example, $v_1v_3\in\Eleftedge$, $v_2v_5\in\Erightedge$ and $v_5v_7\in\Ebelowedge$ in~\autoref{fig:representation} (a).

Given two points $p$ and $q$, we define $\ell(p,q)$ as the line segment connecting them. Given an edge $e$, we define $\area(e)$ as the region enclosed in the closed loop (including the boundary) defined by the union of the edge $e=v_iv_j$ and the line segment $\ell(v_j,v_i)$. (Note that $\ell(v_j,v_i)$ lies on $\ellfront$.) This notion of area induces a natural partial order ``$\leq$'' on the edges of $G$. $$e_1\leq e_2\,\Longleftrightarrow\,\area(e_1)\subseteq\area(e_2).$$ For example, $v_2v_6\leq v_2v_7$ and $v_4v_2\leq v_5v_7$ in~\autoref{fig:representation} (a). It is easy to see that $\leq$ is reflexive, anti-symmetric and transitive. Thus, $(E,\leq)$ is a poset. Consider the Haase diagram of this poset, where the minimal elements are placed at the bottom and the maximal elements at the top. For an edge $e$, let $\rank(e)$ be the number of elements (including $e$) on the longest downward chain starting from $e$. For example, $\rank(v_2v_7)=3$ because $v_2v_5\leq v_2v_6\leq v_2v_7$ is the longest downward chain starting from $v_2v_7$ in~\autoref{fig:representation} (a). Note that the $\rank$ of all the minimal elements is one. We are now all set to construct $C$, our $\PureTwoDir$ representation of $\apexgraph{G}$ for $k=3$ (\autoref{fig:representation} (d)).

%\newpage
\begin{mdframed}[style=MyFrame,backgroundcolor=blue!1!white]
	
	\medskip\noindent\textbf{\underline{The apex vertex $a$:}} Let $\curve{a}$ be the vertical segment $((0,0.5),(0,n+0.5))$.
	
	\medskip\noindent\textbf{\underline{The vertices $\{v_1,v_2,\ldots,v_n\}$:}} For each $i\in[n]$, let $\curve{v_i}$ be the horizontal segment $((-a_i-0.1,i),(b_i+0.1,i))$, where $a_i$ and $b_i$ are defined as follows.
	\begin{align}
	a_i&=\max(\{0\}\cup\{\rank(e)\mid e \text{ is incident to $v_i$ and $e$ connects to $v_i$ from the left of $\ellfront$}\});\label{eq:ai}\\
	b_i&=\max(\{0\}\cup\{\rank(e)\mid e\text{ is incident to $v_i$ and $e$ connects to $v_i$ from the right of $\ellfront$}\}).\label{eq:bi}
	\end{align}
	The $\{0\}$ set is included to ensure that the argument for the $\max$ operation is not an empty set.
	
	\medskip\noindent\textbf{\underline{The vertices $\{u^1_e,u^2_e,u^3_e\}$:}} For each edge $e$ of $G$, we define a set of four points $\ell_e=(\alpha_e,\beta_e,\gamma_e,\delta_e)$, such that
	\begin{align*}
	\curve{u^1_e}&=\ell(\alpha_e,\beta_e);\\
	\curve{u^2_e}&=\ell(\beta_e,\gamma_e);\\
	\curve{u^3_e}&=\ell(\gamma_e,\delta_e).
	\end{align*}
	We may think of $\ell_e$ as a piecewise linear curve with three pieces.
	Let $\zeta:E\rightarrow\{1,2,\ldots,n^2\}$ be an injective function\footnote{One way to formulate such a function is by setting $\zeta(e)=ni+j$ for each $e=v_iv_j$, where $1\leq i<j\leq n$.}, \ie, $\zeta$ maps each edge $e$ of $G$ to a distinct number from the set $\{1,2,\ldots,n^2\}$. For each edge $e$, let
	\begin{equation} \label{eq:xposrank}
	\xpos(e)=\rank(e)+\frac{\zeta(e)}{n^4}.
	\end{equation}
	Let $\varepsilon_k=1/(k^2n^5)$ (for this construction, $k=3$). We use $\xpos$ and $\varepsilon_k$ to define the four points of $\ell_e$, for each $e$.%We are now set to define $\ell_e$.
	\begin{align*}
	\ell_e =
	\begin{cases} 
	(\left(-\xpos(e),i\right),\left(-\xpos(e),\frac{i+j}{2}\right),\left(-\xpos(e)-\varepsilon_k,\frac{i+j}{2}\right),\left(-\xpos(e)-\varepsilon_k,j\right)) &\text{if } e\in\Eleftedge; \\
	(\left(\xpos(e),i\right),\left(\xpos(e),\frac{i+j}{2}\right),\left(\xpos(e)+\varepsilon_k,\frac{i+j}{2}\right),\left(\xpos(e)+\varepsilon_k,j\right)) &\text{if } e\in\Erightedge; \\
	(\left(-\xpos(e),i\right),\left(-\xpos(e),n+\rank(e)\right),\left(\xpos(e),n+\rank(e)\right),\left(\xpos(e),j\right)) &\text{if } e\in\Eaboveedge; \\
	(\left(-\xpos(e),i\right),\left(-\xpos(e),-\rank(e)\right),\left(\xpos(e),-\rank(e)\right),\left(\xpos(e),j\right)) &\text{if } e\in\Ebelowedge.
	\end{cases}
	\end{align*}
\end{mdframed}
Let $\Grep{G}$ be the intersection graph of $C$. We need to show that $\apexgraph{G}=\Grep{G}$. First, let us understand the idea behind our construction of $C$.

We consider each $\ell_e$ a single (piecewise linear) segment. Note that in all four cases of its definition, $\ell_e$ always consists of two vertical segments $\left(\curve{u^1_e}, \curve{u^3_e}\right)$ and one horizontal segment $\left(\curve{u^2_e}\right)$. Also, the x-coordinate of the vertical segments of $\ell_e$ is essentially the $\rank$ (or the negation of the $\rank$) of $e$ (\autoref{eq:xposrank}). The $\zeta(e)/n^4$ term (and also the $\varepsilon_k$ term) is simply a tiny perturbation added to its x-coordinate to ensure that the vertical parts of $\ell_e$ do not intersect the vertical parts of any other $\ell_{e'}$. ($\zeta$ was chosen to be an injection for precisely this reason.)

Since $\apexgraph{G}$ and $\Grep{G}$ have the same vertex set, it is sufficient to show that $e\in E(\apexgraph{G})\Leftrightarrow e\in E(\Grep{G})$ in order to demonstrate their equality.

\medskip\noindent\textbf{Proof of $e\in E(\apexgraph{G})\Rightarrow e\in E(\Grep{G})$:} The $\curve{v_i}$'s are horizontal segments, all intersecting the vertical apex segment $\curve{a}$. Further, the $\curve{v_i}$'s are made to extend as far (to the left and right of $\curve{a}$) as the maximum $\rank$ (plus an additional $\pm 0.1$) of their incident edges (\autoref{eq:ai},~\autoref{eq:bi}). This ensures that they intersect the vertical segments of all their corresponding $\ell_e$'s. The fact that the $\curve{u^1_e}$'s and $\curve{u^3_e}$'s intersect their corresponding $\curve{u^2_e}$'s is implicit from the definition of $\ell_e$.

\medskip\noindent\textbf{Proof of $e\notin E(\apexgraph{G})\Rightarrow e\notin E(\Grep{G})$:} Note that $C$ has three types of segments: (i) the apex ($\curve{a}$), (ii) the horizontal segments $\curve{v_i}$, and (iii) the piecewise linear segments $\ell_e$. We will consider all pairs of non-adjacent vertices $(p,q)$ of $\apexgraph{G}$, and show that $\curve{p}$ and $\curve{q}$ do not intersect in $C$. We have three cases.

\medskip\noindent\textit{\underline{Case 1: one of $\curve{p}$ or $\curve{q}$ is of type (i).}} Let us say $\curve{p}$ is of type (i), \ie, $p$ is the apex vertex $a$. Then $\curve{p}=((0,0.5),(0,n+0.5))$, and $\curve{q}$ must be of type (iii) (since all type (ii) vertices are adjacent to $a$). Note that the x-coordinates of the vertical pieces of all the $\ell_e$'s in $\Eleftedge\cup\Erightedge$ are either less than $-0.1$ or greater than $0.1$, and the y-coordinates of the horizontal pieces of all the $\ell_e$'s in $\Eaboveedge\cup\Ebelowedge$ are either greater than $n+0.5$ or less than $0.5$. Therefore, $\curve{p}$ intersects none of the $\ell_e$'s.

\medskip\noindent\textit{\underline{Case 2: one of $\curve{p}$ or $\curve{q}$ is of type (ii).}} Let us say $\curve{p}$ is of type (ii). If $\curve{q}$ is also of type (ii), then we are done, since all the $\curve{v_i}$'s are mutually disjoint (\autoref{eq:ai},~\autoref{eq:bi}). If $\curve{q}$ is of type (iii), then let $\curve{q}$ be a piece of $\ell_e$ for some edge $e$, and let $p=v_k$ for some $k\in[n]$ such that $e$ is not incident to $v_k$.

If $v_k\notin\area(e)$, then all segments of $\ell_e$ lie above $\curve{v_k}$ or all segments of $\ell_e$ lie below $\curve{v_k}$, implying that $\ell_e$ and $\curve{v_k}$ do not intersect. If $v_k\in\area(e)$, then $\area(e)$ lies on at least one of the sides (left/right) of $v_k\in\ellfront$. Let $e'$ be an edge of maximum $\rank$ incident to $v_k$ such that $e'$ connects to $v_k$ from the same side (left/right) of $\ellfront$ as $\area(e)$. (If no such edge exists, then the $\{0\}$ set (\autoref{eq:ai},~\autoref{eq:bi}) comes into play, and we are done, as $\curve{v_k}$ falls short of $\ell_e$.) Recall that $\area(e)$ comprises of $e$ and the front line $\ellfront$. Using the facts that $e'$ does not cross $e$ or $\ellfront$ and that $e'$ connects to $v_k$ from the interior of $\area(e)$, we obtain that $e'$ lies entirely inside $\area(e)$. Therefore, $$\area(e')\subseteq\area(e)\ \Rightarrow\ e'\leq e\ \Rightarrow\ \rank(e')<\rank(e)\ \Rightarrow\ \rank(e')+1\leq\rank(e).$$ Note that the vertical pieces of $\ell_e$ are at least $\rank(e)$ units away from the apex segment $\curve{a}$, and the horizontal segment $\curve{v_k}$ only reaches as far as $\rank(e')+0.1<\rank(e)$ units from $\curve{a}$ in the direction of $\ell_e$. Thus, $\curve{v_k}$ and $\ell_e$ do not intersect. Also, all the horizontal pieces of $\ell_e$ of non-$\varepsilon_k$ length belong to edges of $\Ecrossedge$, which lie above $\curve{a}$ or below $\curve{a}$, and thus none of them intersect $\curve{v_k}$.

\medskip\noindent\textit{\underline{Case 3: both $\curve{p}$ and $\curve{q}$ are of type (iii).}} Let $\curve{p}$ be a piece of $\ell_{e_p}$ and $\curve{q}$ be a piece of $\ell_{e_q}$, for some edges $e_p$ and $e_q$ of $G$. If they are both vertical pieces or one of them is a horizontal piece of length $\varepsilon_k$, then the $\zeta$ function guarantees that they do not intersect.

Thus, the only remaining case is if one of them (say $\curve{p}$) is a horizontal piece of non-$\varepsilon_k$ length (\ie, $e_p$ is a crossover edge). Let $e_p\in\Eaboveedge$ (the proof for $e_p\in\Ebelowedge$ is similar). The y-coordinate of $\curve{p}$ is greater than $n+0.5$. If $e_q\in\Eleftedge\cup\Erightedge\cup\Ebelowedge$, then all segments of $\ell_{e_q}$ lie below $\curve{p}$, and we are done. If $e_q\in\Eaboveedge$, then note that the all the edges contained in $\Eaboveedge$ constitute a total order (or chain) in the poset $(E,\leq)$. Thus either $\rank(e_p)<\rank(e_q)$ or $\rank(e_q)<\rank(e_p)$. Let $\rank(e_p)<\rank(e_q)$ (the proof for $\rank(e_q)<\rank(e_p)$ is similar). All the pieces of $\ell_{e_q}$ are at least $\rank(e_p)+0.5$ units away from the apex segment $\curve{a}$, and all the pieces of $\ell_{e_q}$ reach less than $\rank(e_p)+0.1$ units away from $\curve{a}$. Hence, $\ell_{e_p}$ and $\ell_{e_q}$ do not intersect.

This completes the proof of the base case ($k=3$) of our induction. A crucial feature of our construction, which we will exploit in our proof of the inductive case, is that for all edges $e$ of $G$, the segment $\curve{u^3_e}$ is a vertical segment.

\paragraph{Induction hypothesis:} Let $k\geq 3$ be an odd integer. Then there exists a $\PureTwoDir$ representation of $\apexgraph{G}$ in which $\curve{u^k_e}$ is a vertical segment for all edges $e$ of $G$.

\paragraph{Induction step:} Given a $\PureTwoDir$ representation of $\apexgraph{G}$ where $\curve{u^k_e}$ is a vertical segment, we will slightly modify it to include two new segments $\curve{u^{k+1}_e}$ and $\curve{u^{k+2}_e}$ for each $e$, such that $\curve{u^{k+2}_e}$ is a vertical segment. Let
\begin{align*}
\curve{u^k_e}&=((\lambda^k_e,\mu^k_e),(\lambda^k_e,\pi^k_e));\\
\sigma^{k+2}_e&=
\begin{cases}
-1 \qquad\text{ if } \lambda^k_e<0;\\
+1 \qquad\text{ if } \lambda^k_e>0.
\end{cases}
\end{align*}
Recall that $\varepsilon_k=1/(k^2n^5)$. Now for each edge $e$ of $G$, we replace the segment $\curve{u^k_e}$ by the following three segments.
\begin{align*}
\curve{u^k_e}&=\displaystyle{\left(\left(\lambda^k_e,\mu^k_e\right),\left(\lambda^k_e,\frac{\mu^k_e+\pi^k_e}{2}\right)\right)};\\
\curve{u^{k+1}_e}&=\displaystyle{\left(\left(\lambda^k_e,\frac{\mu^k_e+\pi^k_e}{2}\right),\left(\lambda^k_e+\sigma^{k+2}_e\varepsilon_{k+2},\frac{\mu^k_e+\pi^k_e}{2}\right)\right)};\\
\curve{u^{k+2}_e}&=\displaystyle{\left(\left(\lambda^k_e+\sigma^{k+2}_e\varepsilon_{k+2},\frac{\mu^k_e+\pi^k_e}{2}\right),\left(\lambda^k_e+\sigma^{k+2}_e\varepsilon_{k+2},\pi^k_e\right)\right)}.
\end{align*}
Note that these three new segments roughly coincide with the segment that they replaced, with a tiny perturbation of $\sigma^{k+2}_e\varepsilon_{k+2}$ made to the x-coordinates of $\curve{u^{k+1}_e}$ and $\curve{u^{k+2}_e}$. Using the induction hypothesis, it is easy to see that $\curve{u^k_e}$, $\curve{u^{k+1}_e}$ and $\curve{u^{k+2}_e}$ intersect the segments that they are adjacent to in $\apexgraph{G}$. The following calculation shows that the $\sigma^{k+2}_e\varepsilon_{k+2}$ perturbation is so minuscule that $\curve{u^{k+1}_e}$ and $\curve{u^{k+2}_e}$ do not intersect any additional segments.
\begin{equation*}
\displaystyle{\left|\sum_{i=3}^{k+2}\sigma^i_e\varepsilon_i\right| \leq \frac{1}{n^5}\left( \sum_{i=3}^{k+2}\frac{1}{i^2}\right) < \frac{1}{n^5}}.
\end{equation*}
Note that for every $e'\neq e$ and every odd $k'$ such that $1\leq k'\leq k+2$, the x-coordinates of $\curve{u^{k'}_{e'}}$ and $\curve{u^k_e}$ differ by roughly $1/n^4$, which is much larger than $1/n^5$. Finally, note that $\curve{u^{k+2}_e}$ is a vertical segment, as promised. This completes the proof.
\end{proof}

\section{Planarizable Representations of Planar Graphs} \label{sec:planarrep}

In this section, we will show~\autoref{lem:regstringrepplanar}, \ie, a graph admits a planarizable representation (\autoref{def:planar}) if and only if the graph is planar.

\begin{proof}[Proof of~\autoref{lem:regstringrepplanar}] It is easy to see that every planar graph admits a planarizable representation. We will show the other direction: every graph that admits a planarizable representation is planar. Let $G$ be a graph with a planarizable representation. Let $v$ be a vertex of $G$, and $\curve{v}$ be its corresponding string. \autoref{fig:planplan} shows the steps of our proof for a given $\curve{v}$. Let $$R_v=\{p \mid p\in\mathbb{R}^2, d(p,\curve{v})\leq\varepsilon\}$$ be the set of points within a closed $\varepsilon$-neighbourhood of $\curve{v}$, choosing $\varepsilon$ small enough so that $\curve{v}$ does not intersect any additional strings. Delete all substrings lying in the interior of $R_v$. Thus all strings that intersected $\curve{v}$ now have one end point on the boundary of $R_v$. Connect all these boundary end points to a common point (say $p_v$) in the interior of $R_v$ via pairwise disjoint substrings (intersecting only at $p_v$) in the interior of $R_v$, effectively ``shrinking'' the region $R_v$ to a single point $p_v$. (This last step is possible because $R_v$ is a simply connected region.) Now the point $p_v$ corresponds to the vertex $v$.
	
	Do this for all the vertices of $G$. Since the vertices are now points, and the edges are strings connecting them, the representation thus obtained is a planar drawing of $G$.
\end{proof}

%\begin{comment}

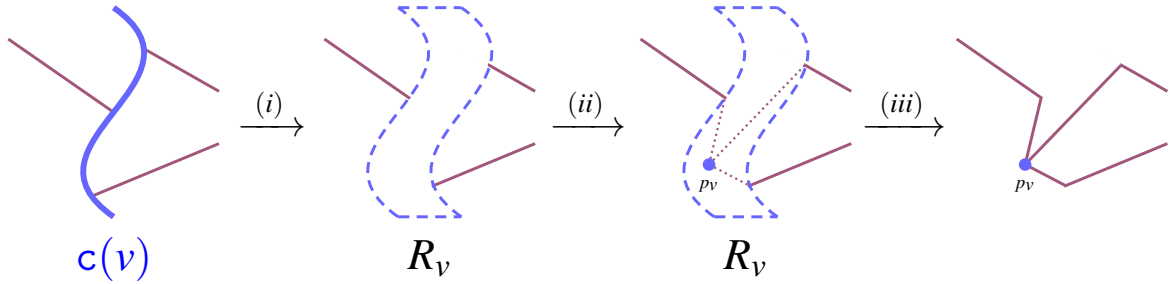
\begin{figure}
	\begin{centering}
		\resizebox{1\textwidth}{!}{
			\begin{tikzpicture}
			
			%\draw [ultra thick, magenta!70!black] (-5,4) .. controls (0,5) and (5,3) ..  (11,4);
			
			\foreach \i in {-4,-1,...,5}
			{
				\draw [thick, magenta!60!black] (\i,4.7) -- (\i+1,4);
				\draw [thick, magenta!60!black] (\i+2,3.7) -- (\i+0.8,3.2);
				\draw [thick, magenta!60!black] (\i+2,4.2) -- (\i+1.3,4.6);
			}
			
			\foreach \i in {-0.25,-0.2,...,0.3}
			\draw [ultra thick, white] (\i,3) .. controls (\i-1,3.7) and (\i+1,4.3) .. (\i,5);
			
			\draw [ultra thick, blue!60!white] (-3,3) .. controls (-4,3.7) and (-2,4.3) .. (-3,5);
			
			\draw [thick, densely dashed, blue!60!white] (-0.3,3) .. controls (-1.3,3.7) and (0.7,4.3) .. (-0.3,5);
			\draw [thick, densely dashed, blue!60!white] (-0.3,3) -- (0.3,3);
			
			\draw [thick, densely dashed, blue!60!white] (0.3,3) .. controls (-0.7,3.7) and (1.3,4.3) .. (0.3,5);
			\draw [thick, densely dashed, blue!60!white] (-0.3,5) -- (0.3,5);
			
			\foreach \i in {2.75,2.8,...,3.3}
			\draw [ultra thick, white] (\i,3) .. controls (\i-1,3.7) and (\i+1,4.3) .. (\i,5);
			
			\draw [thick, densely dashed, blue!60!white] (2.7,3) .. controls (1.7,3.7) and (3.7,4.3) .. (2.7,5);
			\draw [thick, densely dashed, blue!60!white] (2.7,3) -- (3.3,3);
			
			\draw [thick, densely dashed, blue!60!white] (3.3,3) .. controls (2.3,3.7) and (4.3,4.3) .. (3.3,5);
			\draw [thick, densely dashed, blue!60!white] (2.7,5) -- (3.3,5);
			
			\def \pvx{2.65};
			\def \pvy{3.5};
			
			\draw [semithick, densely dotted, magenta!60!black] (\pvx,\pvy) -- (2.8,4.13);
			
			\draw [semithick, densely dotted, magenta!60!black] (\pvx,\pvy) -- (3.56,4.45);
			
			\draw [semithick, densely dotted, magenta!60!black] (\pvx,\pvy) -- (3.03,3.3);
			
			\filldraw[blue!60!white] (\pvx,\pvy) circle (1.5pt);
			\node at (\pvx,\pvy-0.2) {\tiny $p_v$};
			
			\foreach \i in {5.75,5.8,...,6.3}
			\draw [ultra thick, white] (\i,3) .. controls (\i-1,3.7) and (\i+1,4.3) .. (\i,5);
			
			\def \pvx{5.65};
			
			\draw [thick, magenta!60!black] (\pvx,\pvy) -- (5.8,4.13);
			\draw [thick, magenta!60!black] (\pvx,\pvy) -- (6.56,4.45);
			\draw [thick, magenta!60!black] (\pvx,\pvy) -- (6.03,3.3);
			
			\filldraw[blue!60!white] (\pvx,\pvy) circle (1.5pt);
			\node at (\pvx,\pvy-0.2) {\tiny $p_v$};
			
			\node[blue] at (-3,2.6) {\large $\curve{v}$};
			\node at (0,2.6) {\large $R_v$};
			\node at (3,2.6) {\large $R_v$};
			
			\node[black] at (-1.5,4) {$\xrightarrow{\ (i)\ }$};
			\node[black] at (1.5,4) {$\xrightarrow{\ (ii)\ }$};
			\node[black] at (4.5,4) {$\xrightarrow{\ (iii)\ }$};
			
			\end{tikzpicture}
		}
		\caption{Initially, the vertex $v$ is denoted by the blue string $\curve{v}$, and the edges incident to it are denoted in red. (i) $\curve{v}$ is ``thickened'' to form a region $R_v$ around it. (ii) The end points of the edges on the boundary of $R_v$ are connected to a single point $p_v$ in the interior of $R_v$. (iii) Strings that share an end point on the boundary of $R_v$ are concatenated, and the region $R_v$ is ``shrunk'' to the point $p_v$.
		}
		\label{fig:planplan}
	\end{centering}
\end{figure}

\section{Conclusion}\label{sec:conclude}
\autoref{cor:box} states that recognizing rectangle intersection graphs is $\NP$-hard, even when the inputs are bipartite apex graphs. This raises the following question. Can we recognize planar rectangle intersection graphs in polynomial time? As mentioned earlier, all planar graphs are intersection graphs of 3-dimensional hyper-rectangles (or axis-parallel cuboids)~\cite{thomassen1986interval,felsner2011}. Furthermore, there exist series-parallel graphs (a subclass of planar graphs) that are not rectangle intersection graphs~\cite{bohra2006}.

\autoref{cor:fpt} states that for all graph classes $\cG$ such that $\PureTwoDir\subseteq\cG\subseteq\OneString$, the recognition of $\cG$ is not FPT (fixed-parameter tractable), when parameterized by the apex number of the graph. As our construction produces graphs of large degree, the maximum degree of the graph might be a parameter for which the recognition of $\cG$ is FPT.

\autoref{cor:minorfree} states that recognizing several geometric intersection graph classes is $\NP$-hard, even when the inputs are restricted to $K_6$-minor free graphs. On the other hand, the complexity of finding geometric representations of $K_5$-minor free graphs is unknown. Is it possible to use Wagner's Theorem~\cite{wagner1937eigenschaft} to decide in polynomial time whether a $K_5$-minor free graph is in $\OneString$ (or $\textsc{String}$)? It would also be interesting to study the complexity of recognizing $\String$ when the inputs are restricted to apex graphs.

The crossing number of a graph is the minimum number of edge crossings possible in a plane drawing of the graph. Planar graphs are precisely the graphs with crossing number zero. Schaefer showed that apex graphs can have arbitrarily high crossing number, and also exhibited several graphs with crossing number one~\cite{schaefer2018crossing}. Graph classes with a small crossing number, like $k$-planar graphs~\cite{angelini2020simple}, have also been studied. Therefore, the complexity of recognizing $\OneString$ (or $\String$) when the inputs are restricted to graphs with a small crossing number is another potential direction of research.

Finally, it would be interesting to see if our techniques can be used to prove $\NP$-hardness of recognizing other classes of geometric intersection graphs, like outerstring graphs~\cite{biedl2018} and intersection graphs of grounded \textsc{L}-shapes~\cite{mcguinness1996}. Also, the graph classes we study in this paper are for objects embedded in the plane. The complexity of finding geometric intersection representations of apex graphs (appropriately defined) using curves on other surfaces (\eg, torus, projective plane) is another avenue open for exploration.

\medskip\noindent\textbf{Acknowledgements:} This project has received funding from the European Union’s Horizon 2020 research and innovation programme under grant agreement No. 682203-ERC-[Inf-Speed-Tradeoff]. The authors thank the organisers of \textsc{Graphmasters 2020}~\cite{gkasieniec2020} for providing the virtual environment that initiated this research.

%\medskip\noindent\textbf{Acknowledgements:} This project has received funding from the European Union’s Horizon 2020 research and innovation programme under grant agreement No. 682203-ERC-[Inf-Speed-Tradeoff]. The authors thank the organisers of \textsc{Graphmasters 2020} for providing the virtual environment to initiate the research.
 
\bibliographystyle{alpha}
\bibliography{references}

\end{document}